\documentclass[a4paper]{article}
\usepackage{amscd,amsmath,amssymb}
\usepackage[utf8]{inputenc}
\usepackage{xcolor}
\usepackage{float}
\usepackage{tikz}
\usepackage{tikz-feynman}
\usetikzlibrary{decorations.pathreplacing}
\usetikzlibrary{calc}
\usepackage{amsthm}
\usepackage{ytableau}

\newcommand{\Cl}{\mathcal{C}l}

\newcommand{\hC}{\widehat{C}}
\newcommand{\ad}{{\sf ad}}

\newcommand{\kg}{\mathsf{g}}
\newcommand{\mcU}{\mathcal{U}}

\newcommand{\oDelta}{\overline{\Delta}}

\newcommand{\Spin}{\mathrm{Spin}}

\newcommand{\hR}{\hat{R}}
\newcommand{\hsR}{\hat{\mathsf{R}}}

\newcommand{\tikznode}[3][inner sep=0pt]{\tikz[remember
picture,baseline=(#2.base)]{\node(#2)[#1]{$#3$};}}

\newcommand{\be}{\begin{equation}}
\newcommand{\ee}{\end{equation}}
\newcommand{\bea}{\begin{eqnarray}}
\newcommand{\eea}{\end{eqnarray}}

\newcommand{\ba}{\begin{array}} \newcommand{\ea}{\end{array}}

\DeclareMathOperator{\proj}{P}

\DeclareMathOperator{\tr}{tr}

\DeclareMathOperator{\diag}{diag}

\newcommand{\okg}{\overline{\mathsf{g}}}

\newcommand{\mfg}{\mathfrak{g}}

\allowdisplaybreaks

\numberwithin{equation}{section}

\newtheorem{proposition}{Proposition}
\newtheorem{lemma}{Lemma}
\theoremstyle{remark}
\newtheorem{remark}{Remark}

\ytableausetup{centertableaux}

\begin{document}
\begin{flushright}
\today \\
\end{flushright}
\vspace{2.3cm}

\begin{center}
{\huge\bf Split Casimir Operator of the Lie Algebra $so_{2r}$ in Spinor Representations, Colour Factors, and the Yang–Baxter Equation}
\end{center}
\vspace{1cm}

\begin{center}
{\Large \bf A. P. Isaev${}^{a,b}$ and A. A. Provorov${}^{a}$}
\end{center}

\vspace{0.2cm}

\begin{center}
{${}^a$ \it
Bogoliubov  Laboratory of Theoretical Physics,\\
Joint Institute for Nuclear Research,
141980 Dubna, Russia}\vspace{0.1cm}

{${}^b$ \it Lomonosov Moscow State University, \\
 Physics Faculty, Russia}\vspace{0.1cm}

{\tt isaevap@theor.jinr.ru, aleksanderprovorov@gmail.com}
\end{center}
\vspace{3cm}

\begin{abstract}
In this paper, we derive characteristic identities for the split Casimir operator of the Lie algebra $so_{2r}$ in tensor products of spinor representations of the same and opposite chiralities. Using these identities, we explicitly construct projectors onto invariant subspaces of this operator and compute their traces. The results obtained allow us to derive explicit expressions for the colour factors of ladder Feynman diagrams in gauge theories with gauge group $\Spin(2r)$.In addition, we obtain a new form of a solution to the Yang–Baxter equation that is invariant under the action of the Lie algebra $so_{2r}$ in spinor representations.
\end{abstract}

\vspace{1cm}

\newpage

\pagenumbering{arabic}
\setcounter{page}{2}

\section{Introduction}
The split (more generally, $n$-split) Casimir operator plays an important role in the theory of Lie algebras and Lie superalgebras, as well as in representation theory. In particular, a number of works~\cite{IsKrad2,IsPrad2,IsKrPrad3,IsPrad3,Prad3,IsKrMkAv,IsKrad5} established universal characteristic identities for the 2-, 3-, 4-, and 5-split Casimir operators. These results made it possible to construct explicitly the projectors onto universal subrepresentations of tensor powers of the adjoint representation of simple Lie algebras and superalgebras. In~\cite{IsVUB}, similar calculations were carried out for the 2-split Casimir operator acting in the tensor products $T\otimes Y_n$ and $T\otimes Y_n'$, where $T$ denotes the fundamental representation and $Y_n$, $Y_n'$ are universal representations of simple Lie algebras. Moreover, the split Casimir operator admits an interpretation in terms of the colour factors of Feynman diagrams in non-Abelian gauge theories~\cite{IsPr:Colour} (see also~\cite{Cvit}), which makes it a convenient tool for calculating scattering amplitudes in perturbation theory.

The problem of calculating colour factors occupies a central position in modern particle physics. The Standard Model, which currently provides the most precise theoretical description of the fundamental interactions, is a non-Abelian gauge theory with gauge group $SU(3) \times SU(2) \times U(1)$. Consequently, the computation of colour factors and the development of methods for their evaluation play a crucial role in the calculation of scattering amplitudes and other physical quantities.

Apart from the Standard Model, grand unified theories (GUTs) are of considerable interest, as they provide a unified description of the strong, weak, and electromagnetic interactions at high energies. The first such theory was proposed in~\cite{GeorGlas}, where it was shown that the fermion fields of the Standard Model can be unified into the $\overline{5}$ and $10$ multiplets of $SU(5)$. In a subsequent work~\cite{FritMink}, the possibility of constructing a grand unified theory with gauge group $SO(10)$ (more precisely, $\Spin(10)$) was considered. In this framework, all quarks and leptons of a single generation are accommodated within a single 16-dimensional spinor multiplet. Despite difficulties in explaining the hierarchy of scales associated with the breaking of $SO(10)$ symmetry down to the Standard Model symmetry, this theory has proven rather successful. In particular, it was shown in~\cite{G-MRamSlan} that the presence in the spinor multiplet of an additional particle, the right-handed neutrino absent from the $\overline{5}$ and $10$ multiplets of $SU(5)$, makes it possible to construct mechanisms explaining the origin of neutrino masses and their smallness. Moreover,~\cite{BaMoh} demonstrated the strong predictive power of the minimal $SO(10)$ model in describing the neutrino spectrum. Other popular grand unified theories are based on the gauge groups $E_6$ and $E_8$~\cite{E_6 GUT, E_8 GUT}.

Another important area of application of colour factors in quantum field theory is the so-called $1/N$ expansion introduced by 't~Hooft in~\cite{tHooft}. He showed that the structure of an $SU(N)$ gauge theory is drastically simplified in the limit $N \to \infty$, since in this case only planar Feynman diagrams survive in the computation of various amplitudes. This led 't~Hooft to formulate a method of calculations in quantum chromodynamics known as the $1/N$ expansion~\cite{tHooft}. This method had a significant impact on the development of quantum field theory and string theory (see the review~\cite{AGMO}); it found applications in the analysis of confinement~\cite{tHooft2}, properties of baryons~\cite{Witten}, and matrix models~\cite{FrGinZinn}.

The split Casimir operator also has applications in the field of quantum integrable systems. It is used in constructing solutions to the quantum Yang--Baxter equation that are invariant under the action of Lie algebras and Lie superalgebras~\cite{IsKrad2,IsPrad2,MacKay}. This equation first appeared in the works of McGuire~\cite{McGuire} and Yang~\cite{Yang} and plays a crucial role in the theory of quantum integrability~\cite{Baxter,Zam} (see also~\cite{ChaPr} and references therein). In particular, within the framework of the quantum inverse scattering method~\cite{SkTFad}, certain structures emerged that eventually led to the development of the theory of quantum groups~\cite{Drin,Jim,Jim2}, which are deformations of Lie groups and Lie algebras and describe the symmetries of quantum integrable systems~\cite{PasSal,KarZap}.

For exceptional Lie algebras, $R$-matrices were examined in the setting of the Freudenthal–Tits magic square by Westbury~\cite{West}, with connections to Vogel’s parametrisation.

The paper is organised as follows. In Section~2 we introduce the split Casimir operator of a simple Lie algebra and discuss its main properties. Sections~3 and~4 present two alternative approaches to deriving characteristic identities for the split Casimir operator in tensor products of spinor representations. Section~5 contains an example of the computation of colour factors for Feynman diagrams in a gauge theory with gauge group $\Spin(2r)$, while Section~6 is devoted to the construction of solutions to the Yang--Baxter equation that are invariant under the action of the Lie algebra $so_{2r}$ in spinor representations. The Conclusion briefly summarises the main results of the paper.

\section{Basic definitions}
\subsection{The split Casimir operator of a simple Lie algebra}

In this section, we recall some well-known facts about the quadratic and split Casimir operators of simple Lie algebras (see, e.g., \cite{IsKrad2,IsKrPrad3,IsRu1}).

Let $\mfg$ be a simple complex Lie algebra with basis $\{X_A\}$ and structure relations
\begin{equation}
[X_A,X_B]=X^{C}{}_{AB} X_C,
\end{equation}
where $X^{C}{}_{AB}$ are the structure constants. The Cartan--Killing metric $\kg_{AB}:=X^{C}{}_{AD}X^{D}{}_{BC}$ of $\mfg$ and the inverse metric $\okg^{AB}$, satisfying $\okg^{AB}\kg_{BC}=\delta^A_C$, allow one to define the quadratic Casimir operator
\begin{equation}\label{C_2 definition}
C_2:=\okg^{AB}X_AX_B\in\mcU(\mfg),
\end{equation}
which is a central element of the universal enveloping algebra $\mcU(\mfg)$. Consequently, in any representation $T$ of $\mfg$ it acts as an invariant operator: $[T(X_A),T(C_2)]=0$ for any $X_A$.

Let $V_\lambda$ be the space of an irreducible representation $T_\lambda$ of $\mfg$ with highest weight $\lambda$. Then the operator $T_\lambda(C_2)$ is proportional to the identity operator $I_{V_\lambda}$ on $V_\lambda$:
\begin{equation}
T_\lambda(C_2) = c_2^{\lambda} I_{V_\lambda},
\end{equation}
where $c_2^{T_\lambda} \equiv c_2^{\lambda}$ is the eigenvalue of $T_\lambda(C_2)$ in $V_\lambda$. For the adjoint representation $T_\lambda = \ad$, it follows immediately from~\eqref{C_2 definition} that $c_2^{\ad} = 1$. In the general case,
\begin{equation}\label{c_2 via weights}
c_2^{\lambda} = (\lambda, \lambda + 2\delta),
\end{equation}
where $\delta$ denotes the Weyl vector of $\mfg$, and the scalar product $(\ ,\ )$ in the root space is normalised by the condition $c_2^{\ad} = 1$ obtained above. For example, in the case of the Lie algebra $so_N$, which is the main object of study in this paper, this means that for arbitrary basis elements $e^{(i)}$ and $e^{(j)}$ of its root space one has
\begin{equation}\label{Root product norm}
(e^{(i)}, e^{(j)}) = \frac{1}{2(N-2)} \delta^{ij}.
\end{equation}

We now define the split Casimir operator of $\mfg$ by
\begin{equation}\label{CasDef}
\hC=\okg^{AB}X_A\otimes X_B\ \in\  \mcU(\mfg)\otimes \mcU(\mfg).
\end{equation}
It is related to the quadratic Casimir operator $C_2$ by
\begin{equation}\label{C_2 and C_(2) relation}
\Delta C_2=C_2\otimes I+I\otimes C_2+2\hC
\iff
\hC=\frac{1}{2}\bigl(\Delta C_2-C_2\otimes I-I\otimes C_2\bigr),
\end{equation}
where $I\in\mcU(\mfg)$ denotes the identity element and
$\Delta:\mcU(\mfg)\to \mcU(\mfg)\otimes \mcU(\mfg)$ is the comultiplication,
\begin{equation}
\Delta(I)=I\otimes I,\qquad
\Delta(X_A)=X_A\otimes I+I\otimes X_A.
\end{equation}
The operator $\hC$ is $\ad$-invariant. More precisely, for arbitrary representations
$T$ and $T'$ of $\mfg$ one has
\begin{equation}
\label{adinv}
[T(X_A)\otimes I_{T'}+I_T\otimes T'(X_A),\,\hC_{T\cdot T'}]=0
\qquad \forall X_A,
\end{equation}
where $I_T$ denotes the identity operator in the representation $T$, and
$\hC_{T\cdot T'}:=(T\otimes T')\hC$.

It follows from~\eqref{C_2 and C_(2) relation} that the eigenvalues
$c_{(2)}^{\lambda}$ and $c_{2}^{\lambda}$ of the split Casimir operator
$\hC$ and the quadratic Casimir operator $C_2$ of $\mfg$ in the representation
$T_\lambda$, appearing in the decomposition
$T\otimes T'=\bigoplus_{\lambda}T_\lambda$, are related by
\begin{equation}\label{Split via quadratic}
c_{(2)}^{\lambda}=\frac{1}{2}\bigl(c_{2}^{\lambda}-c_2^T-c_2^{T'}\bigr).
\end{equation}

The $\ad$-invariance of $\hC_{T\cdot T'}$ together with Schur's lemma implies that if
$T\otimes T'$ decomposes as a direct sum of irreducible representations as $T\otimes T'=\bigoplus_{\lambda}T_\lambda$, then the space $V_{\lambda}$ of each subrepresentation $T_\lambda$ is an eigenspace of $\hC_{T\cdot T'}$ with eigenvalue $c_{(2)}^{\lambda}$, and the following identity holds:
\begin{equation}\label{CharDef}
\prod_{c_{(2)}^{\lambda}}\bigl(\hC_{T\cdot T'}-c_{(2)}^{\lambda}\bigr)=0,
\end{equation}
where the product is taken over all distinct eigenvalues $c_{(2)}^{\lambda}$.
This identity is called the characteristic identity of $\hC_{T\cdot T'}$.
By construction, \eqref{CharDef} has the minimum possible degree among all polynomial relations satisfied by this operator.

Using \eqref{CharDef}, one constructs a system of orthogonal projectors
$\proj_{c_{(2)}^{\lambda}}$ onto the eigenspaces $\{V_{c_{(2)}^{\lambda}}\}$ of $\hC_{T\cdot T'}$\footnote{The spaces $V_{c_{(2)}^{\lambda}}$ are called Casimir eigenspaces of the representation $T\otimes T'$ and are, in general, direct sums of the spaces of irreducible representations $T_\lambda,\,T_{\lambda'},\dots$ in which $\hC_{T\cdot T'}$ acts with the same eigenvalue $c_{(2)}^{\lambda}=c_{(2)}^{\lambda'}=\dots$.}:
\begin{equation}\label{GenPr}
\proj_{c_{(2)}^{\lambda}}=
\prod_{c_{(2)}^{\rho}\neq c_{(2)}^{\lambda}}
\frac{\hC_{T\cdot T'}-c_{(2)}^{\rho}}
{c_{(2)}^{\lambda}-c_{(2)}^{\rho}}.
\end{equation}
These projectors can in turn be used to compute the dimensions $\dim V_{c_{(2)}^\lambda}=\tr\proj_{c_{(2)}^\lambda}$ of the corresponding eigenspaces and to express the operator $\hC_{T\cdot T'}$ in the form
\begin{equation}\label{Split Casimir expansion}
\hC_{T\cdot T'}=
\sum_{c_{(2)}^{\lambda}}
c_{(2)}^{\lambda}\,\proj_{c_{(2)}^{\lambda}},
\end{equation}
from which a simple formula for an arbitrary power $L$ of $\hC_{T\cdot T'}$ follows:
\begin{equation}\label{Split Casimir to the L expansion}
\hC_{T\cdot T'}^L=
\sum_{c_{(2)}^{\lambda}}
\bigl(c_{(2)}^{\lambda}\bigr)^L\,
\proj_{c_{(2)}^{\lambda}}.
\end{equation}
In~\eqref{Split Casimir expansion} and~\eqref{Split Casimir to the L expansion}, the summation is taken over all distinct eigenvalues $c_{(2)}^{\lambda}$.

The split Casimir operator admits a graphical interpretation in terms of Feynman diagrams in a non-Abelian gauge theory with Lie algebra $\mfg$. More precisely, its components $(\hC_{T\cdot T})^{i_1 i_2}{}_{j_1 j_2}$ in the representation $T\otimes T$ coincide with the colour factor of the diagram shown in Fig.~\ref{Split Casimir graphic representation}. This diagram corresponds to the interaction of two particles whose fields transform under the action of $\mfg$ in the representation $T$, by exchanging a gauge boson. For a theory with gauge group $\Spin(2r)$, this interpretation will be discussed in greater detail in Section~\ref{Sec: diagrams} (see also~\cite{IsPr:Colour}).

\begin{figure}
\centering
\begin{tikzpicture}

\begin{feynman}
\vertex(i2){$i_2$};
\vertex[below of=i2](i1){$i_1$};
\vertex[right of=i2,label=above:$T_A$](a2);
\vertex[right of=i1,label=below:$T_B$](a1);
\vertex[right of=a2](j2){$j_2$};
\vertex[right of=a1](j1){$j_1$};

\diagram{
(i2)--[fermion](a2)--[fermion](j2),
(i1)--[fermion](a1)--[fermion](j1),
(a1)--[boson, edge label'=$\okg^{AB}$](a2)
};
\end{feynman}

\node[above left of=i1,xshift=-1cm]{$(\hC_{T\cdot T})^{i_1i_2}_{\ j_1j_2}=$};
\end{tikzpicture}
\caption{Graphical interpretation of the operator $\hC_{T\cdot T}$. In the case of a self-dual representation $T$, the horizontal lines should be regarded as unoriented.}
\label{Split Casimir graphic representation}
\end{figure}

We also note that the split Casimir operator is used in constructing solutions to the quantum and quasiclassical Yang--Baxter equations (see, e.g., \cite{MacKay,ChaPr,IsQGLect}). This topic will be discussed in detail in Section~\ref{Sec: YBE} in connection with solutions to the quantum Yang--Baxter equation that are invariant under the action of the Lie algebra $so_{2r}$ in spinor representations.

\section{First approach to the derivation of characteristic identities for the split Casimir operator of the Lie algebra $so_{2r}$ in tensor products of spinor representations}

The necessary information about the Lie algebra $so_N$, the Clifford algebra $\Cl_N$, and their representations is collected in Appendix~\ref{app: so_N and Cl_N props}. For brevity, we omit explicit reference to the irreducible representation $\rho$ of $\Cl_{2r}$ in the formulas of this section and write $\Gamma_i$ instead of $\rho(\Gamma_i)$, $M_{ij}$ instead of $\rho(M_{ij})$, and so on.

\subsection{Characteristic identity of the split Casimir operator of the Lie algebra $so_{2r}$ in the representation $\rho\otimes \rho$}\label{Sec: 3.1}

We introduce elements $I_k$ of the algebra $\rho(\Cl_{2r})\otimes \rho(\Cl_{2r})$, $k=0,1,2,\dots$, defined by (see~\cite{IsKarKir})
\begin{equation}\label{Ik Def}
I_0:=I\otimes I,\qquad
I_k:=\Gamma^{[i_1\dots i_k]}\otimes \Gamma_{[i_1\dots i_k]},\quad k>0.
\end{equation}
It is straightforward to verify that they are invariant under the adjoint action of $so_N$, that is, they satisfy
\begin{equation}
[M_{ij}\otimes I+I\otimes M_{ij},I_k]=0.
\end{equation}
In particular,
\begin{equation}\label{I2 via hC}
I_2=\Gamma^{[i_1 i_2]}\otimes \Gamma_{[i_1 i_2]}
=-16(N-2)\hC_\rho,
\end{equation}
where we have used \eqref{hC via Gammas} and denoted $\hC_\rho:=(\rho\otimes \rho)\hC$. The proportionality between $I_2$ and the split Casimir operator $\hC_\rho$ of $so_{2r}$ in the representation $\rho\otimes \rho$ plays a key role in what follows.

In~\cite{IsKarKir}, a recurrence relation for the elements $I_k$ was obtained:
\begin{equation}\label{IkI1 relation}
I_k I_1 = I_{k+1} - k\bigl((k-1)-2r\bigr) I_{k-1}.
\end{equation}
By successive application of this relation, each invariant $I_k$ can be expressed as a polynomial in $I_1$ of degree $k$, for example:
\begin{equation}\label{I_2, I_3 via I_1 example}
I_2 = I_1^2 - 2r I_0,
\qquad
I_3 = I_1^3 - 2(3r-1) I_1.
\end{equation}

In what follows, we will need an analogue of~\eqref{IkI1 relation} for the even invariants $I_{2k}$:
\begin{equation}\label{I2kI2 relation}
I_{2k} I_2 = I_{2k+2}
+ 8k(r-k) I_{2k}
+ 4k(2k-1)(r+1-k)(2r+1-2k) I_{2k-2}.
\end{equation}
This relation is obtained by multiplying~\eqref{IkI1 relation} by $I_1$, substituting $k \mapsto 2k$, and eliminating $I_1$ using the first equality in~\eqref{I_2, I_3 via I_1 example}. It follows from~\eqref{I2 via hC} and~\eqref{I2kI2 relation} that each invariant $I_{2k}$ can be expressed as a polynomial in $\hC_\rho$ of degree $k$, for example:
\begin{equation}\label{I4, I6 example}
\begin{aligned}
I_4(\hC_\rho) &= 1024(r-1)^2 \hC_\rho^2
+ 256(r-1)^2 \hC_\rho
- 4r(2r-1) I_0,\\
I_6(\hC_\rho) &= -32768(r-1)^3 \hC_\rho^3
- 8192(r-1)^2(3r-5)\hC_\rho^2
- 128(r-1)(18r^2-65r+46)\hC_\rho \\
&\quad + 64r(r-2)(2r-1) I_0,\\
I_8(\hC_\rho) &= 2^{20}(r-1)^4 \hC_\rho^4
+ 2^{19}(r-1)^3(3r-7)\hC_\rho^3
+ 2^{13}(r-1)^2(66r^2-301r+308)\hC_\rho^2 \\
&\quad + 2^{11}(r-1)(10r^3-91r^2+217r-132)\hC_\rho
- 48r(r-2)(2r-1)(22r-71) I_0,\\
I_{10}(\hC_\rho) &= -2^{25}(r-1)^5 \hC_\rho^5
- 2^{24}\cdot 5(r-1)^4(r-3)\hC_\rho^4
- 2^{18}(r-1)^3(230r^2-1335r+1806)\hC_\rho^3 \\
&\quad - 2^{16}(r-1)^2(190 r^3-1665 r^2+4473 r-3590)\hC_\rho^2 \\
&\quad - 2^9(r-1)(140 r^4-5820 r^3+36351 r^2-72610 r+40536)\hC_\rho \\
&\quad + 2^9(r-2) r (2 r-9)(2 r-1)(19 r-62) I_0,\\
I_{12}(\hC_\rho) &= 2^{30}(r-1)^6 \hC_\rho^6
+ 2^{28}\cdot 5(r-1)^5(3 r-11)\hC_\rho^5
+ 2^{22}(r-1)^4(1170 r^2-8305 r+13992)\hC_\rho^4 \\
&\quad + 2^{21}(r-1)^3(1070 r^3-11165 r^2+36663 r-37400)\hC_\rho^3 \\
&\quad + 2^{14}(r-1)^2(18660 r^4-269060 r^3+1354221 r^2-2778930 r+1914616)\hC_\rho^2 \\
&\quad - 2^{12}(r-1)(1404 r^5-2200 r^4-105897 r^3+607695 r^2-1108426 r+585720)\hC_\rho \\
&\quad - 320 r(r-2)(2 r-9)(2 r-1)(622 r^2-5755 r+12172).
\end{aligned}
\end{equation}

Due to the antisymmetry of the basis elements $\Gamma_{[i_1\dots i_k]}$ of $\Cl_{2r}$ under permutations of the indices $i_1,\dots,i_k$, all the invariants $I_k$ with $k>2r$ vanish: $I_{2r+1}=I_{2r+2}=I_{2r+3}=\dots=0$. The first even invariant of this type that can be written as a polynomial in $\hC_\rho$ is $I_{2r+2}$. Accordingly, the following proposition holds:

\begin{proposition}\label{I_k's, Prop: char identities for hC_rho}
The split Casimir operator of $so_{2r}$ in the representation $\rho$ satisfies
\begin{equation}\label{I2 char iden}
I_{2r+2}(\hC_\rho)=0.
\end{equation}
Identity~\eqref{I2 char iden} is the characteristic identity of
$\hC_\rho$.
\end{proposition}

\begin{remark}
Strictly speaking, the minimality of the degree of the polynomial
$I_{2r+2}(\hC_\rho)$ does not follow directly from the arguments given above.
This part of Proposition~\ref{I_k's, Prop: char identities for hC_rho} will be
proved later in Section~\ref{Sec: hC_ee' char idens second derivation} by an
alternative method.
\end{remark}

Since $\tr I_{2m}=0$ for $m\ge 1$, as follows from the definition~\eqref{Ik Def}, the successive application of the trace to identities~\eqref{I2kI2 relation} for
$k=1,2,3,\dots$ makes it possible to compute the values of
$\tr(\hC_\rho^{k+1})$. In particular,
\begin{equation}
\begin{aligned}
\tr \hC_\rho^2 &= \frac{r(2r-1)}{256(r-1)^2}\tr I_0,
& \tr \hC_\rho^3 &= -\frac{r(2r-1)}{1024(r-1)^2}\tr I_0,\\
\tr \hC_\rho^4 &= \frac{r(2r-1)(30r^2-63r+34)}{2^{16}(r-1)^4}\tr I_0,
& \tr \hC_\rho^5 &= -\frac{r(2r-1)(34r^2-89r+62)}{2^{17}(r-1)^4}\tr I_0.
\end{aligned}
\end{equation}
Since the projectors~\eqref{GenPr} onto invariant subspaces of $\hC_\rho$ are polynomials in $\hC_\rho$,
computing $\tr(\hC_\rho^k)$ allows one to determine the traces
$\tr \proj_{c_{(2)}^\lambda}$ of these projectors, which coincide with the dimensions of the corresponding subspaces.

As an illustration, we present the characteristic identities~\eqref{I2 char iden} for the operator $\hC_\rho$ in the cases of the Lie algebras $so_4$, $so_6$, $so_8$, and $so_{10}$, together with the dimensions of the corresponding invariant subspaces:
\begin{equation}
so_4:\qquad
\begin{gathered}
I_6|_{r=2}\sim\hC_\rho\Big(\hC_\rho-\frac{1}{8}\Big)\Big(\hC_\rho+\frac{3}{8}\Big)=0,\\ \dim V_0=8,\quad \dim V_{\frac{1}{8}}=6,\quad \dim V_{-\frac{3}{8}}=2,
\end{gathered}
\end{equation}
\begin{equation}
so_6:\qquad
\begin{gathered}
I_8|_{r=3}\sim\Big(\hC_\rho-\frac{1}{32}\Big)\Big(\hC_\rho-\frac{3}{32}\Big)\Big(\hC_\rho+\frac{5}{32}\Big)\Big(\hC_\rho+\frac{15}{32}\Big)=0,\\
\dim V_{\frac{1}{32}}=30,\quad \dim V_{\frac{3}{32}}=20,\quad \dim V_{-\frac{5}{32}}=12,\quad \dim V_{-\frac{15}{32}}=2.
\end{gathered}
\end{equation}
\begin{equation}
so_{8}:\qquad
\begin{gathered}
I_{10}|_{r=4}\sim \Big(\hC_\rho-\frac{1}{12}\Big)\Big(\hC_\rho-\frac{1}{24}\Big)\Big(\hC_\rho+\frac{1}{12}\Big)\Big(\hC_\rho+\frac{7}{24}\Big)\Big(\hC_\rho+\frac{7}{12}\Big)=0,\\
\dim V_{\frac{1}{12}}=70,\quad \dim V_{\frac{1}{24}}=112,\quad \dim V_{-\frac{1}{12}}=56,\quad \dim V_{-\frac{7}{24}}=16,\quad \dim V_{-\frac{7}{12}}=2.
\end{gathered}
\end{equation}
\begin{equation}
so_{10}:\qquad
\begin{gathered}
I_{12}|_{r=5}\sim \Big(\hC_\rho-\frac{5}{64}\Big)\Big(\hC_\rho-\frac{3}{64}\Big)\Big(\hC_\rho+\frac{3}{64}\Big)\Big(\hC_\rho+\frac{13}{64}\Big)\Big(\hC_\rho+\frac{27}{64}\Big)\Big(\hC_\rho+\frac{45}{64}\Big)=0,\\
\dim V_{\frac{5}{64}}=252,\quad \dim V_{\frac{3}{64}}=420,\quad \dim V_{-\frac{3}{64}}=240,\\
\dim V_{-\frac{13}{64}}=90,\quad \dim V_{-\frac{27}{64}}=20,\quad \dim V_{-\frac{45}{64}}=2.
\end{gathered}
\end{equation}

\subsection{Characteristic identities of the split Casimir operator of the Lie algebra $so_{2r}$ in the representations $\Delta_{\pm}\otimes \Delta_{\pm}$ and $\Delta_{\pm}\otimes \Delta_{\mp}$}

To derive the characteristic identities of the split Casimir operator in the tensor product of two spinor representations\footnote{The spinor representations of $so_{2r}$ of positive and negative chirality are denoted by $\Delta_+$ and $\Delta_-$, respectively; see Appendix~\ref{app: so_N and Cl_N props}.}
of the Lie algebra $so_{2r}$, we define the restriction of $\hC_\rho$ to the representation
$\Delta_{\epsilon}\otimes \Delta_{\epsilon'}$, $\epsilon,\epsilon'=\pm$,
by means of the projectors $\proj_\epsilon$ introduced in~\eqref{rho to Delta_pm projectors}:
\begin{equation}\label{hC_epsilon epsilon' def}
\hC_{\epsilon\epsilon'}:=
(\proj_\epsilon\otimes \proj_{\epsilon'})\hC_\rho
\equiv
\proj_{\epsilon\epsilon'}\hC_\rho,
\end{equation}
where we have set $\proj_{\epsilon\epsilon'}:=\proj_\epsilon\otimes \proj_{\epsilon'}$.

Since the projectors $\proj_+$ and $\proj_-$ are mutually orthogonal and invariant under the action of $so_{2r}$ in the representation $\rho$, one has $\proj_{\epsilon\epsilon'}\hC_\rho^k = (\proj_{\epsilon\epsilon'}\hC_\rho)^k = \hC_{\epsilon\epsilon'}^k$, so that for any polynomial $J(\hC_\rho)$, $\proj_{\epsilon\epsilon'} J(\hC_\rho) = J(\hC_{\epsilon\epsilon'})$. Multiplying~\eqref{I2 char iden} by $\proj_{\epsilon\epsilon'}$ therefore yields
\begin{equation}\label{I2 in Deltas init iden}
I_{2r+2}(\hC_{\epsilon\epsilon'})=0.
\end{equation}
However, identity~\eqref{I2 in Deltas init iden} is not the characteristic identity of $\hC_{\epsilon\epsilon'}$, since it is not of minimal degree. Indeed, compared to the case of the representation $\rho$, the polynomials $I_{2k}(\hC_{\epsilon\epsilon'})$ satisfy additional algebraic relations. To derive these relations, we will need the following lemma.

\begin{lemma}\label{Lemma 1}
The invariants $I_k$ satisfy
\begin{equation}\label{(I otimes Gamma)I_k}
(I\otimes \Gamma_{2r+1})\frac{I_k}{k!}
=
(-1)^r(\Gamma_{2r+1}\otimes I)
\frac{I_{2r-k}}{(2r-k)!}.
\end{equation}
\end{lemma}

\begin{proof}
The basis elements $\Gamma_{[i_1\dots i_k]}$ of $\Cl_{2r}$ satisfy the following identity upon multiplication by $\Gamma_{2r+1}$ (see, e.g.,~\cite{IsRu2}):
\begin{equation}\label{Gamma_N+1 product identity}
\Gamma_{[i_1\dots i_k]}\Gamma_{2r+1}
=
(-i)^r(-1)^{[\frac{k}{2}]}
\frac{1}{(2r-k)!}
\varepsilon_{i_1\dots i_k i_{k+1}\dots i_{2r}}
\Gamma^{[i_{k+1}\dots i_{2r}]},
\end{equation}
where $[\frac{k}{2}]$ denotes the integer part of $\frac{k}{2}$, and
$\varepsilon_{i_1\dots i_{2r}}$ is the totally antisymmetric tensor with
$\varepsilon_{1,2,\dots,2r}=1$.

Using~\eqref{Gamma_N+1 product identity}, the left-hand side of~\eqref{(I otimes Gamma)I_k} can be rewritten as
\begin{equation}
\begin{aligned}
(I\otimes \Gamma_{2r+1})\frac{I_k}{k!} &= \frac{1}{k!} \Gamma^{[i_1\dots i_k]} \otimes \Gamma_{2r+1}\Gamma_{[i_1\dots i_k]} \\
&= \frac{1}{k!(2r-k)!} (-i)^r(-1)^{[\frac{k}{2}] + k} \varepsilon_{i_1\dots i_{2r}} \Gamma^{[i_1\dots i_k]} \otimes \Gamma^{[i_{k+1}\dots i_{2r}]} \\
&= \frac{1}{k!(2r-k)!} (-i)^r(-1)^{[\frac{k}{2}] + k + k(2r-k)} \varepsilon_{i_{k+1}\dots i_{2r} i_1\dots i_k} \Gamma^{[i_1\dots i_k]} \otimes \Gamma^{[i_{k+1}\dots i_{2r}]} \\
&= \frac{1}{(2r-k)!} (-1)^{[\frac{k}{2}] + k + k(2r-k) + [\frac{2r-k}{2}] + 2r - k} \Gamma_{2r+1} \Gamma_{[i_{k+1}\dots i_{2r}]} \otimes \Gamma^{[i_{k+1}\dots i_{2r}]}.
\end{aligned}
\end{equation}
Using the identity
\begin{equation}
(-1)^{[\frac{k}{2}] + [\frac{2r-k}{2}]} = (-1)^{r+k},
\end{equation}
we obtain~\eqref{(I otimes Gamma)I_k}.
\end{proof}

Since the projectors $\proj_{\epsilon\epsilon'}$ are constructed from the operators
$I\otimes \Gamma_{2r+1}$ and $\Gamma_{2r+1}\otimes I$, Lemma~\ref{Lemma 1}
allows one to prove the following.

\begin{proposition}
The polynomials $I_{2k}(\hC_{\epsilon\epsilon'})$ satisfy
\begin{equation}\label{I_k's: hC possible idens}
I_{2r-2k}(\hC_{\epsilon\epsilon'})-\epsilon\epsilon' \frac{(2r-2k)!}{(2k)!}I_{2k}(\hC_{\epsilon\epsilon'})=0,\quad k=0,1,\dots,2r.
\end{equation}
\end{proposition}
\begin{proof}
For any $k=0,1,\dots,2r$, consider the following chain of equalities:
\begin{equation}\label{proj_ee' I_k}
\begin{aligned}
\proj_{\epsilon\epsilon'}\frac{I_k}{k!}&=\frac{1}{4}\big((I+\epsilon \Gamma_{{2r}+1})\otimes (I+\epsilon' \Gamma_{{2r}+1})\big)\frac{I_k}{k!}\\
&=\frac{1}{4k!}(I\otimes I)I_k+\frac{1}{4k!}\epsilon(\Gamma_{{2r}+1}\otimes I)I_k+\frac{1}{4k!}\epsilon'(I\otimes \Gamma_{{2r}+1})I_k+\frac{1}{4k!}\epsilon\epsilon'(\Gamma_{{2r}+1}\otimes \Gamma_{{2r}+1})I_k\\
&=\frac{(-1)^r}{4({2r}-k)!}(\Gamma_{{2r}+1}\otimes \Gamma_{{2r}+1})I_{{2r}-k}+\frac{(-1)^r}{4({2r}-k)!}\epsilon (I\otimes \Gamma_{{2r}+1})I_{{2r}-k}+\\
&+\frac{(-1)^r}{4({2r}-k)!}\epsilon'(\Gamma_{{2r}+1}\otimes I)I_{{2r}-k}+\frac{(-1)^r}{4({2r}-k)!}\epsilon\epsilon' (I\otimes I)I_{{2r}-k}\\
&=(-1)^r\frac{1}{4}\big((\epsilon I+\Gamma_{{2r}+1})\otimes (\epsilon' I+\Gamma_{{2r}+1})\big)\frac{I_{{2r}-k}}{({2r}-k)!}=\epsilon\epsilon' (-1)^r\proj_{\epsilon\epsilon'}\frac{I_{{2r}-k}}{({2r}-k)!}.
\end{aligned}
\end{equation}
Here we have used the involutive property~\eqref{Gamma_N+1 properties}
of $\Gamma_{2r+1}$ and the multiplication rule~\eqref{(I otimes Gamma)I_k}
for $(I\otimes \Gamma_{2r+1})$ and $(\Gamma_{2r+1}\otimes I)$ acting on $I_k$.
Then~\eqref{I_k's: hC possible idens} follows from~\eqref{proj_ee' I_k} by substituting $k\mapsto 2k$ and expressing the invariants
$\proj_{\epsilon\epsilon'} I_{2k}$ and $\proj_{\epsilon\epsilon'} I_{2r-2k}$
as polynomials in $\hC_{\epsilon\epsilon'}$
using~\eqref{I2 via hC} and~\eqref{I2kI2 relation}.
\end{proof}

It is easy to see that the polynomials on the left-hand side of
\eqref{I_k's: hC possible idens} have a smaller degree than the polynomial
in \eqref{I2 in Deltas init iden} for all $k=0,\dots,[\frac{r}{2}]$.
It is therefore natural to assume that the characteristic identity of
$\hC_{\epsilon\epsilon'}$ is given by a relation of the form
\eqref{I_k's: hC possible idens} of minimal possible degree.

\begin{proposition}\label{I_k hC_epsilon,epsilon' char idens}
For even $r$, the operators $\hC_{\epsilon,-\epsilon}$ and
$\hC_{\epsilon,\epsilon}$, $\epsilon=\pm$, satisfy the characteristic
identities
\begin{align}
&\hspace{1.8cm}
I_r(\hC_{\epsilon,-\epsilon})=0, & & I_{r+2}(\hC_{\epsilon\epsilon}) - r(r^2-1)(r+2) I_{r-2}(\hC_{\epsilon\epsilon}) =0 \label{I_k's: hC_(alpha,pm alpha) for even nu} \\
\intertext{of degrees $\frac{r}{2}$ and $\frac{r}{2}+1$, respectively,
and for odd $r$ the identities}
& I_{r+1}(\hC_{\epsilon,-\epsilon}) + r(r+1)I_{r-1}(\hC_{\epsilon,-\epsilon}) =0, & & I_{r+1}(\hC_{\epsilon\epsilon}) - r(r+1)I_{r-1}(\hC_{\epsilon\epsilon}) =0 \label{I_k's: hC_(alpha,pm alpha) for odd nu}
\end{align}
of degree $\frac{r+1}{2}$.
\end{proposition}

\begin{proof}
In the case of even $r$ and $\epsilon=-\epsilon'$, the polynomial on
the left-hand side of \eqref{I_k's: hC possible idens} has the minimum degree
$\tfrac{r}{2}$ achieved at $k=\tfrac{r}{2}$. Substituting this value
into \eqref{I_k's: hC possible idens} yields the first identity in
\eqref{I_k's: hC_(alpha,pm alpha) for even nu}. For $\epsilon=\epsilon'$
the substitution $k=\tfrac{r}{2}$ makes the left-hand side of
\eqref{I_k's: hC possible idens} vanish identically. Hence the minimal
degree in this case is $\tfrac{r}{2}+1$, attained at
$k=\tfrac{r}{2}-1$, which gives the second identity in
\eqref{I_k's: hC_(alpha,pm alpha) for even nu}.

For odd $r$, the polynomial on the left-hand side of
\eqref{I_k's: hC possible idens} has the minimum degree
$\tfrac{r+1}{2}$ attained at $k=\tfrac{r-1}{2}$ independently of the
values of $\epsilon$ and $\epsilon'$. Substituting this value of $k$
into \eqref{I_k's: hC possible idens} yields
\eqref{I_k's: hC_(alpha,pm alpha) for odd nu}.
\end{proof}

\begin{remark}
As in the case of the representation $\rho$, the arguments presented in this section do not prove that for $\hC_{\epsilon\epsilon'}$ there exists no identity of degree lower than that of
\eqref{I_k's: hC_(alpha,pm alpha) for even nu} and
\eqref{I_k's: hC_(alpha,pm alpha) for odd nu}.
Justification of the minimality of the polynomials in
\eqref{I_k's: hC_(alpha,pm alpha) for even nu}
will be given later in Section~\ref{Sec: hC_ee' char idens second derivation}.
\end{remark}

As an illustration, we present the characteristic identities
\eqref{I_k's: hC_(alpha,pm alpha) for even nu} and
\eqref{I_k's: hC_(alpha,pm alpha) for odd nu}
for the operators $\hC_{\epsilon\epsilon'}$ in the cases of the Lie algebras
$so_4$, $so_6$, $so_8$, and $so_{10}$,
together with the dimensions of the corresponding eigenspaces:

\begin{equation}
so_4:\qquad
\begin{aligned}
&\hC_{\epsilon,-\epsilon}=0,& \qquad \qquad &\Big(\hC_{\epsilon\epsilon}-\frac{1}{8}\Big)\Big(\hC_{\epsilon\epsilon}+\frac{3}{8}\Big)=0,\\
&\dim V_0^{\epsilon,-\epsilon}=4,& &\dim V_{\frac{1}{8}}^{\epsilon\epsilon}=3,\qquad \dim V_{-\frac{3}{8}}^{\epsilon\epsilon}=1,
\end{aligned}
\end{equation}
\begin{equation}
so_6:\qquad
\begin{aligned}
&\Big(\hC_{\epsilon,-\epsilon}-\frac{3}{32}\Big)\Big(\hC_{\epsilon,-\epsilon}+\frac{5}{32}\Big)=0, &\qquad& \Big(\hC_{\epsilon\epsilon}-\frac{1}{32}\Big)\Big(\hC_{\epsilon\epsilon}+\frac{15}{32}\Big)=0,\\
&\dim V^{\epsilon,-\epsilon}_{\frac{3}{32}}=10,\quad \dim V^{\epsilon,-\epsilon}_{-\frac{5}{32}}=6,& & \dim V^{\epsilon\epsilon}_{\frac{1}{32}}=15,\quad \dim V^{\epsilon\epsilon}_{-\frac{15}{32}}=1.
\end{aligned}
\end{equation}
\begin{equation}
so_8:\qquad
\begin{aligned}
&\Big(\hC_{\epsilon,-\epsilon}-\frac{1}{24}\Big)\Big(\hC_{\epsilon,-\epsilon}+\frac{7}{24}\Big)=0, &\qquad& \Big(\hC_{\epsilon\epsilon}-\frac{1}{12}\Big)\Big(\hC_{\epsilon\epsilon}+\frac{1}{12}\Big)\Big(\hC_{\epsilon\epsilon}+\frac{7}{12}\Big)=0,\\
&\dim V^{\epsilon,-\epsilon}_{\frac{1}{24}}=56,\quad \dim V^{\epsilon,-\epsilon}_{-\frac{7}{24}}=8,& & \dim V^{\epsilon\epsilon}_{\frac{1}{12}}=35,\quad \dim V^{\epsilon\epsilon}_{-\frac{1}{12}}=28,\quad \dim V^{\epsilon\epsilon}_{-\frac{7}{12}}=1.
\end{aligned}
\end{equation}
\begin{equation}
so_{10}:\quad
\begin{aligned}
&\Big(\hC_{\epsilon,-\epsilon}-\frac{5}{64}\Big)\Big(\hC_{\epsilon,-\epsilon}+\frac{3}{64}\Big)\Big(\hC_{\epsilon,-\epsilon}+\frac{27}{64}\Big)=0, &\quad& \Big(\hC_{\epsilon\epsilon}-\frac{3}{64}\Big)\Big(\hC_{\epsilon\epsilon}+\frac{13}{64}\Big)\Big(\hC_{\epsilon\epsilon}+\frac{45}{64}\Big)=0,\\
&\dim V^{\epsilon,-\epsilon}_{\frac{5}{64}}=126,\quad \dim V^{\epsilon,-\epsilon}_{-\frac{3}{64}}=120,& & \dim V^{\epsilon\epsilon}_{\frac{3}{64}}=210,\quad \dim V^{\epsilon\epsilon}_{-\frac{13}{64}}=45,\\
&\dim V^{\epsilon,-\epsilon}_{-\frac{27}{64}}=10, & &\dim V^{\epsilon\epsilon}_{-\frac{45}{64}}=1.
\end{aligned}
\end{equation}

\section{Second approach to the derivation of characteristic identities for the split Casimir operator of the Lie algebra $so_{2r}$ in tensor products of spinor representations}
\label{Sec: hC_ee' char idens second derivation}

In this section, we use properties of the Lie algebra $so_N$ and its representations, summarised in Appendix~\ref{app: so_N and its external powers}.

\subsection{Characteristic identities of the split Casimir operator of the Lie algebra $so_{2r}$ in the representations $\Delta_{\pm}\otimes \Delta_{\pm}$ and $\Delta_{\pm}\otimes \Delta_{\mp}$}

To derive the characteristic identities of the split Casimir operator
\begin{equation}
\hC=\okg^{i_1i_2,j_1j_2}M_{i_1i_2}\otimes M_{j_1j_2},
\end{equation}
of the Lie algebra $so_{2r}$ (whose inverse Cartan–Killing metric $\okg^{i_1i_2,j_1j_2}$ is given in \eqref{soKillingMetric}) in the representations $\Delta_\epsilon\otimes \Delta_{\epsilon'}$, $\epsilon,\epsilon'=\pm$, that is, for the operators $\hC_{\epsilon\epsilon'}:=(\Delta_\epsilon\otimes \Delta_{\epsilon'})\hC$, we use the following result~\cite{Weyl} (see also~\cite{IsRu2}):

\begin{proposition}\label{Proposition 6.4.1 from IsRu2}
For even $r$ the following decompositions hold:
\begin{equation}\label{Delta times Delta expansion: even nu}
\begin{aligned}
\Delta_\epsilon\otimes \Delta_{-\epsilon}&=T_1\oplus T_3\oplus\dots\oplus T_{r-1},\\
\Delta_\epsilon\otimes \Delta_{\epsilon}&=T_0\oplus T_2\oplus\dots\oplus T_{r-2}+
\left\{\begin{array}{l}
T_r^\epsilon,\quad\text{if}\quad \frac{r}{2}\in 2\mathbb{Z},\\
T_r^{-\epsilon},\quad\text{if}\quad \frac{r}{2}\in 2\mathbb{Z}+1,
\end{array}\right.
\end{aligned}
\end{equation}
and for odd $r$:
\begin{equation}\label{Delta times Delta expansion: odd nu}
\begin{aligned}
\Delta_\epsilon\otimes \Delta_{-\epsilon}&=T_0\oplus T_2\oplus\dots\oplus T_{r-1},\\
\Delta_\epsilon\otimes \Delta_{\epsilon}&=T_1\oplus T_3\oplus\dots\oplus T_{r-2}+
\left\{\begin{array}{l}
T_r^\epsilon,\quad\text{if}\quad \frac{r-1}{2}\in 2\mathbb{Z},\\
T_r^{-\epsilon},\quad\text{if}\quad \frac{r-1}{2}\in 2\mathbb{Z}+1,
\end{array}\right.
\end{aligned}
\end{equation}
where $T_0$ denotes the singlet representation.\footnote{
In~\cite{IsRu2} these decompositions are given for representations of the form
$\Delta_\epsilon\otimes \oDelta_{\epsilon'}$.
They can, however, be brought to the form
\eqref{Delta times Delta expansion: even nu} and
\eqref{Delta times Delta expansion: odd nu}
by using the isomorphisms between $\Delta_\epsilon$ and $\oDelta_{\epsilon'}$
presented in \eqref{Delta bar(Delta) iso}.
}
\end{proposition}

Based on Proposition~\ref{Proposition 6.4.1 from IsRu2}, one obtains the following.

\begin{proposition}\label{Prop: char identities for hC}
For even $r$ the operators $\hC_{\epsilon,-\epsilon}$ and $\hC_{\epsilon,\epsilon}$, $\epsilon=\pm$, satisfy the characteristic identities
\begin{equation}\label{hC_(alpha,pm alpha) for even nu}
\prod_{m=1}^{\frac{r}{2}}\bigg(\hC_{\epsilon,-\epsilon}-c_{(2),2m-1}\bigg)=0,\hspace{2cm}
\prod_{m=0}^{\frac{r}{2}}\bigg(\hC_{\epsilon,\epsilon}-c_{(2),2m}\bigg)=0,
\end{equation}
and for odd $r$:
\begin{equation}\label{hC_(alpha,pm alpha) for odd nu}
\prod_{m=0}^{\frac{r-1}{2}}\bigg(\hC_{\epsilon,-\epsilon}-c_{(2),2m}\bigg)=0,\hspace{2cm}
\prod_{m=1}^{\frac{r+1}{2}}\bigg(\hC_{\epsilon,\epsilon}-c_{(2),2m-1}\bigg)=0,
\end{equation}
where
\begin{equation}\label{hC eigenvalues on Delta times Delta}
c_{(2),k}:=\frac{2k(2r-k)-r(2r-1)}{16(r-1)}.
\end{equation}
\end{proposition}

\begin{proof}
Consider the case of even $r$ and the operator $\hC_{\epsilon,\epsilon}$ (the remaining cases are proved analogously).
Proposition~\ref{Proposition 6.4.1 from IsRu2} implies that the space of the representation $\Delta_\epsilon\otimes \Delta_{\epsilon}$, $\epsilon=\pm$, decomposes into eigenspaces of the quadratic Casimir operator $C_2$ of $so_{2r}$ with eigenvalues given in \eqref{c_2 on T_f^(wedge k)}, where $k$ runs over all even integers from $0$ to $r-2$, together with one additional eigenvalue given in \eqref{c_2 on T_pm^(wedge nu)}.

Using the relation~\eqref{Split via quadratic}, which connects the eigenvalues of $\hC$ and $C_2$ in irreducible representations of $so_{2r}$, and taking into account the eigenvalue~\eqref{c_2 on Delta_pm} of the quadratic Casimir operator in the representation $\Delta_\epsilon$, we obtain expression~\eqref{hC eigenvalues on Delta times Delta} for the spectrum of $\hC$ in the representation $\Delta_\epsilon\otimes \Delta_{\epsilon}$, where $k$ runs over all even integers from $0$ to $2r$. This yields the second identity in \eqref{hC_(alpha,pm alpha) for even nu}.
\end{proof}

The characteristic identities
\eqref{hC_(alpha,pm alpha) for even nu}
and
\eqref{hC_(alpha,pm alpha) for odd nu}
for the operators $\hC_{\epsilon,\epsilon'}$, $\epsilon,\epsilon'=\pm$,
allow one to construct explicit expressions for the projectors
$\proj_k^{\epsilon\epsilon'}\equiv \proj_{c_{(2),k}}^{\epsilon\epsilon'}$
onto their eigenspaces with eigenvalue $c_{(2),k}$ given in
\eqref{hC eigenvalues on Delta times Delta},
using formula~\eqref{GenPr}.

It should be emphasised that for fixed $\epsilon$ and $\epsilon'$
not all values of $k$ occur in the spectrum of $\hC_{\epsilon,\epsilon'}$,
since they are restricted by the decompositions
\eqref{Delta times Delta expansion: even nu}
and
\eqref{Delta times Delta expansion: odd nu}.
For values of $k$ not appearing in these decompositions,
the corresponding projectors
$\proj_k^{\epsilon,\epsilon'}$
vanish identically.

Thus, for even $r$ the non-trivial projectors are
\begin{equation}\label{proj^(alpha,pm alpha) for even nu}
\begin{aligned}
\proj_{2m-1}^{\epsilon,-\epsilon}&=\prod_{\substack{\ell=1\\ \ell\neq m}}^{\frac{r}{2}}\frac{\hC_{\epsilon,-\epsilon}-c_{(2),2\ell-1}}{c_{(2),2m-1}-c_{(2),2\ell-1}},\qquad &\text{for } 1\le m\le \frac{r}{2},\\
\proj_{2m}^{\epsilon,\epsilon}&=\prod_{\substack{\ell=0\\ \ell\neq m}}^{\frac{r}{2}}\frac{\hC_{\epsilon,\epsilon}-c_{(2),2\ell}}{c_{(2),2m}-c_{(2),2\ell}},\qquad &\text{for } 0\le m\le \frac{r}{2},
\end{aligned}
\end{equation}
and for odd $r$:
\begin{equation}\label{proj^(alpha,pm alpha) for odd nu}
\begin{aligned}
\proj_{2m}^{\epsilon,-\epsilon}&=\prod_{\substack{\ell=0\\ \ell\neq m}}^{\frac{r-1}{2}}\frac{\hC_{\epsilon,-\epsilon}-c_{(2),2\ell}}{c_{(2),2m}-c_{(2),2\ell}},\qquad &\text{for } 0\le m\le \frac{r-1}{2},\\
\proj_{2m-1}^{\epsilon,\epsilon}&=\prod_{\substack{\ell=1\\ \ell\neq m}}^{\frac{r+1}{2}}\frac{\hC_{\epsilon,\epsilon}-c_{(2),2\ell-1}}{c_{(2),2m-1}-c_{(2),2\ell-1}},\qquad &\text{for } 1\le m\le \frac{r+1}{2}.
\end{aligned}
\end{equation}

The traces of the projectors \eqref{proj^(alpha,pm alpha) for even nu} and \eqref{proj^(alpha,pm alpha) for odd nu} are given in the following proposition.

\begin{proposition}\label{Prop: traces for projectors}
For even $r$:
\begin{equation}\label{projectors traces for even nu}
\begin{aligned}
\tr \proj_{2m-1}^{\epsilon,-\epsilon}
&=\frac{(2r)!}{(2m-1)!(2r-2m+1)!},
&\qquad &\text{for } 1\le m\le \frac{r}{2},\\
\tr \proj_{2m}^{\epsilon,\epsilon}
&=\frac{(2r)!}{(2m)!(2r-2m)!},
&\qquad &\text{for } 0\le m\le\frac{r}{2}-1,\\
\tr \proj_{r}^{\epsilon,\epsilon}
&=\frac{(2r)!}{2(r!)^2},
\end{aligned}
\end{equation}
and for odd $r$:
\begin{equation}\label{projectors traces for odd nu}
\begin{aligned}
\tr \proj_{2m}^{\epsilon,-\epsilon}
&=\frac{(2r)!}{(2m)!(2r-2m)!},
&\qquad &\text{for } 0\le m\le\frac{r-1}{2},\\
\tr \proj_{2m-1}^{\epsilon,\epsilon}
&=\frac{(2r)!}{(2m-1)!(2r-2m+1)!},
&\qquad &\text{for } 1\le m \le \frac{r-1}{2},\\
\tr \proj_{r}^{\epsilon,\epsilon}
&=\frac{(2r)!}{2(r!)^2}.
\end{aligned}
\end{equation}
\end{proposition}

\begin{proof}
As in the proof of Proposition~\ref{Prop: char identities for hC}, we consider only the case of even $r$ and the projectors $\proj_{2m}^{\epsilon,\epsilon}$; the remaining cases are treated analogously.

The trace of the projector $\proj_{2m}^{\epsilon,\epsilon}$ equals the dimension of the eigenspace of $\hC_{\epsilon,\epsilon}$ with eigenvalue $c_{(2),2m}$. For $m=0,\dots,\frac{r}{2}-1$, this eigenspace coincides with the representation space of $T_{2m}$ whose dimension is given in \eqref{T_f^(wedge k) dimension} with $k=2m$. For $m=\frac{r}{2}$, it coincides with the representation space of $T_r^\pm$ whose dimension is given in \eqref{T_(pm)^(wedge nu) dimension}.
\end{proof}

Consider now the operator $\hC_\rho=\sum_{\epsilon,\epsilon'=\pm}\hC_{\epsilon,\epsilon'}$. For a fixed eigenvalue $c_{(2),k}$, its eigenspace is the direct sum of the eigenspaces of the operators $\hC_{\epsilon,\epsilon'}$, $\epsilon,\epsilon'=\pm$, with the same eigenvalue. The characteristic identities
\eqref{hC_(alpha,pm alpha) for even nu} and
\eqref{hC_(alpha,pm alpha) for odd nu} imply that the operators $\hC_{++}$ and $\hC_{--}$, as well as $\hC_{+-}$ and $\hC_{-+}$, have coinciding spectra. Accordingly, their eigenspaces corresponding to a given eigenvalue combine naturally and are projected onto by
\begin{equation}\label{proj_k^S and proj_k^AS}
\proj_k^{S}:=\proj_k^{++}+\proj_k^{--},\qquad
\proj_k^{AS}:=\proj_k^{+-}+\proj_k^{-+},
\end{equation}
where $\proj_k^{\epsilon,\epsilon'}$ are defined in
\eqref{proj^(alpha,pm alpha) for even nu} and
\eqref{proj^(alpha,pm alpha) for odd nu}.
Clearly, the images of these projectors lie in the spaces of the representations
\begin{equation}\label{Delta^S and Delta^AS def}
\Delta^{S}:=(\Delta_+\otimes \Delta_+)+(\Delta_-\otimes \Delta_-),
\qquad
\Delta^{AS}:=(\Delta_+\otimes \Delta_-)+(\Delta_-\otimes \Delta_+),
\end{equation}
respectively.

Combining Propositions \ref{Prop: char identities for hC} and
\ref{Prop: traces for projectors}, we obtain the following.

\begin{proposition}\label{Prop: char identities for hC in rho otimes rho}
The operator $\hC_\rho$ satisfies the characteristic identity
\begin{equation}\label{hC_rho factorised identity}
\prod_{k=0}^r \big(\hC_\rho-c_{(2),k}\big)=0,
\end{equation}
where $c_{(2),k}$ is defined in
\eqref{hC eigenvalues on Delta times Delta}.

The projector $\proj_k\equiv \proj_{c_{(2),k}}$ onto the eigenspace of $\hC_\rho$ with eigenvalue $c_{(2),k}$ is given, for even $r$, by
\begin{equation}\label{sums of projectors Delta to rho for even r}
\proj_k=\proj_{k}^{S}\quad\text{for}\quad k\in 2\mathbb{Z},
\qquad
\proj_k=\proj_{k}^{AS}\quad\text{for}\quad k\in 2\mathbb{Z}+1,
\end{equation}
and, for odd $r$, by
\begin{equation}\label{sums of projectors Delta to rho for odd r}
\proj_k=\proj_{k}^{AS}\quad\text{for}\quad k\in 2\mathbb{Z},
\qquad
\proj_k=\proj_{k}^{S}\quad\text{for}\quad k\in 2\mathbb{Z}+1,
\end{equation}
where $\proj_k^S$ and $\proj_k^{AS}$ are defined in
\eqref{proj_k^S and proj_k^AS}.

The traces of these projectors are
\begin{equation}
\begin{aligned}
\tr \proj_k &= 2\frac{(2r)!}{k!(2r-k)!}
\qquad \text{for}\quad 0\le k\le r-1,\\
\tr \proj_r &= \frac{(2r)!}{(r!)^2}.
\end{aligned}
\end{equation}
\end{proposition}

\begin{remark}
From Proposition \ref{Prop: char identities for hC in rho otimes rho} and the definition \eqref{proj_k^S and proj_k^AS} of the projectors $\proj_k^{S}$ and $\proj_k^{AS}$ it follows that the images of the projectors $\proj_k$ for $k=0,\dots,r-1$ coincide with the representation spaces of $T_k\oplus T_k$, while the image of the projector $\proj_r$ is the representation space of $T_r=T_r^+\oplus T_r^-$.

Moreover, for even $r$ the images of $\proj_{2m}$ and $\proj_{2m-1}$ lie in the representation spaces of $\Delta^S$ and $\Delta^{AS}$ respectively, whereas for odd $r$ they lie in the representation spaces of $\Delta^{AS}$ and $\Delta^S$, respectively.
\end{remark}

\begin{remark}
Propositions \ref{Prop: char identities for hC} and \ref{Prop: char identities for hC in rho otimes rho} allow one to complete the proofs of Propositions \ref{I_k's, Prop: char identities for hC_rho} and \ref{I_k hC_epsilon,epsilon' char idens}. Since relations \eqref{hC_(alpha,pm alpha) for even nu}, \eqref{hC_(alpha,pm alpha) for odd nu}, and \eqref{hC_rho factorised identity} are the characteristic identities of the corresponding operators, their degrees are minimal. The coincidence of these degrees with those of the identities \eqref{I_k's: hC_(alpha,pm alpha) for even nu}, \eqref{I_k's: hC_(alpha,pm alpha) for odd nu}, and \eqref{I2 char iden} implies that the latter are also the characteristic identities of the operators $\hC_{\epsilon\epsilon'}$ and $\hC_\rho$.
\end{remark}

\section{Colour factors of ladder Feynman diagrams in a gauge theory with fermions transforming in spinor representations of the gauge group $\Spin(2r)$}\label{Sec: diagrams}

To illustrate the application of the split Casimir operator to the calculation of colour factors of Feynman diagrams in non-Abelian gauge theories, we consider a Yang–Mills theory with gauge group $\Spin(2r)$ and fermionic fields $\psi^\alpha$ and $\overline{\psi}_\alpha$,\footnote{Here and in what follows the spinor indices $\alpha$, $\beta,\dots$ refer to the gauge group $\Spin(2r)$ rather than to the group of space–time symmetries.} which transform in its spinor representation $\Delta_+$. The Lagrangian of this theory is given by (see, e.g., \cite{PesSchr})
\begin{equation}\label{Lagrangian}
L=\frac{1}{4}\kg_{i_1i_2,j_1j_2}F^{\mu\nu,i_1i_2}F_{\mu\nu}^{j_1j_2}+\overline{\psi}(i\gamma^\mu D_\mu-m)\psi,
\end{equation}
where the covariant derivative $D_\mu$ and the field-strength tensor $F_{\mu\nu}$ are defined as
\begin{equation}\label{Covariant derivative}
D_\mu=\partial_\mu+ g A_\mu^{ij}\Delta_+(M_{ij}),
\end{equation}
\begin{equation}\label{Field strength tensor}
F_{\mu\nu}^{i_1i_2}=\partial_\mu A_\nu^{i_1i_2}-\partial_\nu A_\mu^{i_1i_2}
+ X^{i_1i_2}{}_{j_1j_2,k_1k_2}A_\mu^{j_1j_2} A_\nu^{k_1k_2}.
\end{equation}
Here $g$ is the coupling constant, $X^{i_1i_2}{}_{j_1j_2,k_1k_2}$ and $\kg_{i_1i_2,j_1j_2}$ denote respectively the structure constants \eqref{soStructureConstants} and the Cartan–Killing metric \eqref{soKillingMetric} of the Lie algebra $so_{2r}$, and the anti-Hermitian\footnote{It is worth noting that in the physics literature the basis elements of the Lie algebra of the gauge group are usually taken to be Hermitian rather than anti-Hermitian. In that case, the second term in \eqref{Covariant derivative} is multiplied by the imaginary unit.} operators $\Delta_+(M_{ij})$ form its basis in the representation $\Delta_+$. Space–time indices are raised and lowered using the Minkowski metric $\eta_{\mu\nu}=\diag(1,-1,-1,-1)$.

In what follows, we shall work in the Feynman gauge. The corresponding Feynman rules (apart from the three- and four-gluon vertices) are shown in Fig.~\ref{FeynmanRules}.
\begin{figure}
\centering
\begin{tikzpicture}
\node(fermion propagator equals) {$=$};
\begin{feynman}
\vertex[left=of fermion propagator equals] (j) {$\beta$};
\vertex [left of=j, xshift=-1cm] (i) {$\alpha$};

\diagram*{
(i) -- [fermion, edge label=$p$] (j),
};
\end{feynman}

\node[right=of fermion propagator equals]{$\displaystyle \delta^\alpha_\beta\frac{i}{p^\mu\gamma_\mu-m+i\varepsilon}$};
\node[below=of fermion propagator equals] (gluon propagator equals) {$=$};

\begin{feynman}
\vertex[left= of gluon propagator equals] (B) {$(j_1,j_2),\nu$};
\vertex[left= of B, xshift=-1.5cm] (A) {$(i_1,i_2),\mu$};

\diagram*{
(A) --[boson, edge label=$p$] (B)
};
\end{feynman}

\node[right=of gluon propagator equals]{$\displaystyle \okg^{i_1i_2,j_1j_2}\frac{i\eta^{\mu\nu}}{p^2+i\varepsilon}$};
\node[below=of gluon propagator equals] (FGF propagator equals) {$=$};

\begin{feynman}
\vertex[left=of FGF propagator equals] (j vertex) {$\beta$};
\vertex[left=of j vertex, xshift=0.2cm] (vertex);
\vertex[left=of vertex, xshift=0.4cm] (k vertex) {$\alpha$};
\vertex[below=of vertex] (A vertex) {$(i_1,i_2)$};

\diagram*{
(k vertex) --[fermion] (vertex) --[fermion] (j vertex),
(vertex) --[boson] (A vertex),
};
\end{feynman}

\node[right = of FGF propagator equals]{$\displaystyle i g \gamma^\mu (S_{i_1i_2})^\alpha{}_\beta$};

\end{tikzpicture}
\caption{Feynman rules for the theory with Lagrangian \eqref{Lagrangian}. If the representations $\Delta_+$ and $\overline{\Delta}_+$ are equivalent, the fermion lines should be treated as unoriented.}\label{FeynmanRules}
\end{figure}
\begin{remark}
The positive sign in front of the kinetic term of the gauge field in the Lagrangian \eqref{Lagrangian} is due to the fact that we employ the negative-definite Cartan–Killing metric $\kg_{i_1i_2,j_1j_2}$, which is convenient for demonstrating the role of the split Casimir operator in the calculation of colour factors.

In general, the Lagrangian is written using the rescaled metric $\kg'_{i_1i_2,j_1j_2}:=T_F\kg_{i_1i_2,j_1j_2}$, where $T_F$ is a constant. Under the replacement $\kg_{i_1i_2,j_1j_2}\mapsto \kg'_{i_1i_2,j_1j_2}$ in the Feynman rules, the colour factor of the gluon propagator is divided by $T_F$, whereas the colour factors of the three- and four-gluon vertices are multiplied by $T_F$.

Accordingly, the correct power $k$ in the overall coefficient $T_F^k$ of an arbitrary Feynman diagram in the theory with metric $\kg'_{i_1i_2,j_1j_2}$ can be recovered from the corresponding diagram in the theory with metric $\kg_{i_1i_2,j_1j_2}$ by counting the number $n_{\text{pr}}$ of gluon propagators and the number $n_{3,4}$ of three- and four-gluon vertices: $k = n_{3,4} - n_{\text{pr}}$.
\end{remark}

As an illustration of the application of the split Casimir operator \eqref{CasDef} to the calculation of colour factors, we consider the components $(\hC_{++})^{\alpha_1\alpha_2}{}_{\beta_1\beta_2}$ of the operator $\hC_{++}:=(\Delta_+\otimes \Delta_+)\hC$ in the tensor product of two spinor representations of positive chirality:
\begin{equation}
\hC_{++}=\okg^{k_1k_2,m_1m_2}\Delta_+(M_{k_1k_2})^{\alpha_1}{}_{\beta_1}\otimes \Delta_+(M_{m_1m_2})^{\alpha_2}{}_{\beta_2}.
\end{equation}
These components coincide with the colour factor of the diagram shown in Fig.~\ref{Split Casimir spinor alpha,alpha graphic representation}.  Accordingly, the components $(\hC_{++}^L)^{\alpha_1\alpha_2}{}_{\beta_1\beta_2}$ of $\hC_{++}^L$ coincide with the colour factor of the ladder diagram shown in Fig.~\ref{fermion-fermion scattering}, which describes the interaction of fermions via the exchange of $L$ gluons.

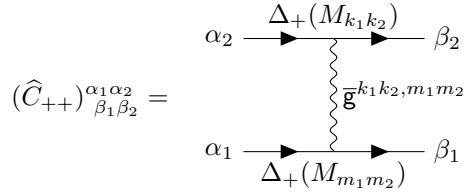
\begin{figure}
\centering
\begin{tikzpicture}

\begin{feynman}
\vertex(i2){$\alpha_2$};
\vertex[below of=i2](i1){$\alpha_1$};
\vertex[right of=i2,label=above:$\Delta_+(M_{k_1k_2})$](a2);
\vertex[right of=i1,label=below:$\Delta_+(M_{m_1m_2})$](a1);
\vertex[right of=a2](j2){$\beta_2$};
\vertex[right of=a1](j1){$\beta_1$};

\diagram{
(i2)--[fermion](a2)--[fermion](j2),
(i1)--[fermion](a1)--[fermion](j1),
(a1)--[boson, edge label'=$\okg^{k_1k_2,m_1m_2}$](a2)
};
\end{feynman}

\node[above left of=i1,xshift=-1cm]{$(\hC_{++})^{\alpha_1\alpha_2}_{\ \beta_1\beta_2}=$};
\end{tikzpicture}
\caption{Graphical interpretation of the operator $\hC_{++}$. For even $r$, the fermion lines should be regarded as unoriented (see the isomorphisms \eqref{Delta bar(Delta) iso}).}
\label{Split Casimir spinor alpha,alpha graphic representation}
\end{figure}
\begin{figure}
\centering
\begin{tikzpicture}
\begin{feynman}
\vertex(i1){$\alpha_2$};
\vertex[below of=i1](i2){$\alpha_1$};
\vertex[right of=i1](k1);
\vertex[right of=i2](k2);
\vertex[right of=k1](k3);
\vertex[right of=k2](k4);
\vertex[right of=k3](k5);
\vertex[right of=k4](k6);
\vertex[right of=k5](k99);
\vertex[right of=k6](k100);
\vertex[right of=k99](k101);
\vertex[right of=k100](k102);
\vertex[right of=k101](j1){$\beta_2$};
\vertex[right of=k102](j2){$\beta_1$};

\diagram*{
(i1)--[fermion](k1),
(i2)--[fermion](k2),
(k1)--[boson](k2),
(k1)--[fermion](k3),
(k2)--[fermion](k4),
(k3)--[boson](k4),
(k3)--[fermion](k5),
(k4)--[fermion](k6),

(k99)--[fermion](k101),
(k100)--[fermion](k102),
(k101)--[boson](k102),
(k101)--[fermion](j1),
(k102)--[fermion](j2)
};

\draw [loosely dotted] (k5)--(k99);
\draw [loosely dotted] (k6)--(k100);
\end{feynman}

\draw[decoration={brace, amplitude=10pt}, decorate] ($(k102.south east)+(10pt,-10pt)$)--($(k2.south west)+(-10pt,-10pt)$) node[pos=0.5,below,yshift=-10pt]{$L$};
\end{tikzpicture}
\caption{Graphical interpretation of the operator $\hC_{++}^L$ in terms of a Feynman diagram describing fermion interaction via gluon exchange.}
\label{fermion-fermion scattering}
\end{figure}
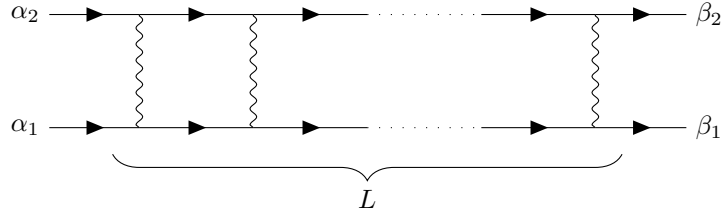

To compute them, we use the characteristic identities of the operator $\hC_{++}$ given on the right-hand sides of \eqref{hC_(alpha,pm alpha) for even nu} and \eqref{hC_(alpha,pm alpha) for odd nu} for even and odd $r$, respectively, together with the expansion \eqref{Split Casimir to the L expansion} of an arbitrary power $L$ of $\hC_{++}$ in terms of the projectors $\proj_{k}^{++}$ onto its invariant subspaces. We thus obtain:
\begin{align}
(\hC_{++}^L)^{\alpha_1\alpha_2}{}_{\beta_1\beta_2}
&=\sum_{m=0}^{\frac{r}{2}} c_{(2),2m}^L (\proj^{++}_{2m})^{\alpha_1\alpha_2}{}_{\beta_1\beta_2}
&& \text{for even } r,\\
(\hC_{++}^L)^{\alpha_1\alpha_2}{}_{\beta_1\beta_2}
&=\sum_{m=1}^{\frac{r+1}{2}} c_{(2),2m-1}^L (\proj^{++}_{2m-1})^{\alpha_1\alpha_2}{}_{\beta_1\beta_2}
&& \text{for odd } r,
\end{align}
where the numbers $c_{(2),k}$ are defined in \eqref{hC eigenvalues on Delta times Delta}.

One possible way of closing the diagram in Fig.~\ref{fermion-fermion scattering} to obtain a vacuum diagram is shown in Fig.~\ref{fermion-fermion scattering vacuum}. The colour factor of this diagram coincides with the trace $(\hC_{++}^L)^{\alpha_1\alpha_2}{}_{\alpha_1\alpha_2} \equiv \tr \hC_{++}^L$, which is computed using expressions \eqref{projectors traces for even nu} and \eqref{projectors traces for odd nu} for the traces of the projectors $\proj^{++}_{k}$:
\begin{equation}
\tr \hC_{++}^L =\sum_{m=0}^{\frac{r}{2}-1} c_{(2),2m}^L \frac{(2r)!}{(2m)!(2r-2m)!} + c_{(2),r}^L \frac{(2r)!}{2(r!)^2} \qquad \text{for even }
\end{equation}
\begin{equation}
\tr \hC_{++}^L =\sum_{m=1}^{\frac{r-1}{2}} c_{(2),2m-1}^L \frac{(2r)!}{(2m-1)!(2r-2m+1)!} + c_{(2),r}^L \frac{(2r)!}{2(r!)^2} \qquad \text{for odd } r,
\end{equation}
where, as before, the coefficients $c_{(2),k}$ are defined in \eqref{hC eigenvalues on Delta times Delta}.
\begin{figure}
\centering
\begin{tikzpicture}
\begin{feynman}
\vertex(i1);
\vertex[below of=i1](i2);
\vertex[right of=i1](k1);
\vertex[right of=i2](k2);
\vertex[right of=k1](k3);
\vertex[right of=k2](k4);
\vertex[right of=k3](k5);
\vertex[right of=k4](k6);
\vertex[right of=k5](k99);
\vertex[right of=k6](k100);
\vertex[right of=k99](k101);
\vertex[right of=k100](k102);
\vertex[right of=k101](j1);
\vertex[right of=k102](j2);

\vertex[above of=i1](i1p);
\vertex[right of=i1p](k1p);
\vertex[right of=k1p](k3p);
\vertex[right of=k3p](k5p);
\vertex[right of=k5p](k99p);
\vertex[right of=k99p](k101p);
\vertex[right of=k101p](j1p);

\vertex[below of=i2](i2p);
\vertex[right of=i2p](k2p);
\vertex[right of=k2p](k4p);
\vertex[right of=k4p](k6p);
\vertex[right of=k6p](k100p);
\vertex[right of=k100p](k102p);
\vertex[right of=k102p](j2p);

\diagram*{
(i1)--[fermion](k1),
(i2)--[fermion](k2),
(k1)--[boson](k2),
(k1)--[fermion](k3),
(k2)--[fermion](k4),
(k3)--[boson](k4),
(k3)--[fermion](k5),
(k4)--[fermion](k6),

(k99)--[fermion](k101),
(k100)--[fermion](k102),
(k101)--[boson](k102),
(k101)--[fermion](j1),
(k102)--[fermion](j2),

(i1p)--[fermion,bend right=90](i1),
(k1p)--[fermion](i1p),
(k3p)--[fermion](k1p),
(k5p)--[fermion](k3p),
(k101p)--[fermion](k99p),
(j1p)--[fermion](k101p),
(j1)--[fermion,bend right=90](j1p),

(i2p)--[fermion,bend left=90](i2),
(k2p)--[fermion](i2p),
(k4p)--[fermion](k2p),
(k6p)--[fermion](k4p),
(k102p)--[fermion](k100p),
(j2p)--[fermion](k102p),
(j2)--[fermion,bend left=90](j2p)
};

\draw [loosely dotted] (k5)--(k99);
\draw [loosely dotted] (k6)--(k100);
\draw [loosely dotted] (k5p)--(k99p);
\draw [loosely dotted] (k6p)--(k100p);
\end{feynman}

\draw[decoration={brace, amplitude=10pt}, decorate] ($(k102.south east)+(10pt,-10pt)$)--($(k2.south west)+(-10pt,-10pt)$) node[pos=0.5,below,yshift=-10pt]{$L$};
\end{tikzpicture}
\caption{Vacuum Feynman diagram obtained by closing the diagram in Fig. \ref{fermion-fermion scattering}.}\label{fermion-fermion scattering vacuum}
\end{figure}
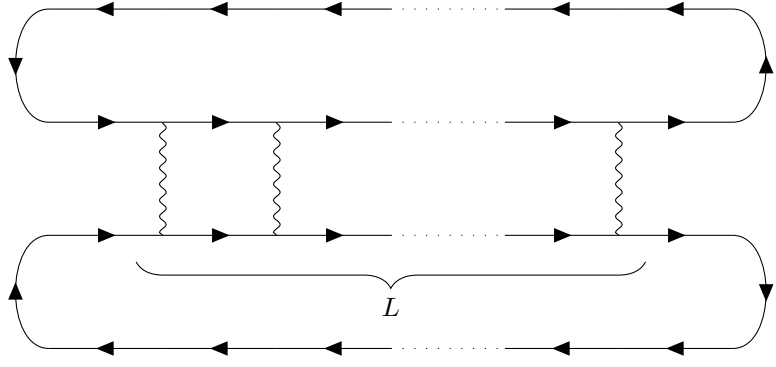

The Feynman diagram shown in Fig.~\ref{fermion-fermion scattering} can also be partially closed by joining, for example, the two upper external lines. The resulting diagram is shown in Fig.~\ref{fermion-fermion scattering semivacuum}. Its colour factor is given by the partial trace $(\hC_{++}^L)^{\alpha_1\alpha_2}{}_{\beta_1\alpha_2} \equiv (\tr_2 \hC_{++}^{L})^{\alpha_1}{}_{\beta_1}$ of the operator $\hC_{++}^L$, where $\tr_2$ denotes the trace over the second tensor component of the representation space $\Delta_+ \otimes \Delta_+$.

To compute this colour factor, it suffices to determine the auxiliary partial traces $\tr_2 \proj_{k}^{++}$. Since these partial traces are invariant operators acting on the space of the irreducible representation $\Delta_+$, each of them is proportional to the identity operator $I_{\Delta_+}$:
\begin{equation}
\tr_2 \proj_k^{++} = a_k^{++} I_{\Delta_+},
\end{equation}
where the coefficients $a_k^{++}$ are to be determined. Taking the full trace of both sides of this equality, we obtain
\begin{equation}
\tr \proj_k^{++}
=
a_k^{++} \tr I_{\Delta_+}
\qquad \implies \qquad
a_k^{++}
=
\frac{\tr \proj_k^{++}}{\tr I_{\Delta_+}}.
\end{equation}

Using expressions \eqref{proj^(alpha,pm alpha) for even nu} and \eqref{proj^(alpha,pm alpha) for odd nu} for $\tr \proj_{k}^{++}$, together with $\tr I_{\Delta_+} = \dim \Delta_+ = 2^{r-1}$, we finally obtain
\begin{equation}
\tr_2 \hC_{++}^L
=
\sum_{m=0}^{\frac{r}{2}-1}
c_{(2),2m}^L
\frac{(2r)!}{2^{r-1}(2m)!(2r-2m)!}
I_{\Delta_+}
+
c_{(2),r}^L
\frac{(2r)!}{2^r (r!)^2}
I_{\Delta_+}
\qquad
\text{for even } r,
\end{equation}
\begin{equation}
\tr_2 \hC_{++}^L
=
\sum_{m=1}^{\frac{r-1}{2}}
c_{(2),2m-1}^L
\frac{(2r)!}{2^{r-1}(2m-1)!(2r-2m+1)!}
I_{\Delta_+}
+
c_{(2),r}^L
\frac{(2r)!}{2^r (r!)^2}
I_{\Delta_+}
\qquad
\text{for odd } r.
\end{equation}

\begin{figure}
\centering
\begin{tikzpicture}
\begin{feynman}
\vertex(i1);
\vertex[below of=i1](i2){$\alpha_1$};
\vertex[right of=i1](k1);
\vertex[right of=i2](k2);
\vertex[right of=k1](k3);
\vertex[right of=k2](k4);
\vertex[right of=k3](k5);
\vertex[right of=k4](k6);
\vertex[right of=k5](k99);
\vertex[right of=k6](k100);
\vertex[right of=k99](k101);
\vertex[right of=k100](k102);
\vertex[right of=k101](j1);
\vertex[right of=k102](j2){$\beta_1$};

\vertex[above of=i1](i1p);
\vertex[right of=i1p](k1p);
\vertex[right of=k1p](k3p);
\vertex[right of=k3p](k5p);
\vertex[right of=k5p](k99p);
\vertex[right of=k99p](k101p);
\vertex[right of=k101p](j1p);

\vertex[below of=i2](i2p);
\vertex[right of=i2p](k2p);
\vertex[right of=k2p](k4p);
\vertex[right of=k4p](k6p);
\vertex[right of=k6p](k100p);
\vertex[right of=k100p](k102p);
\vertex[right of=k102p](j2p);

\diagram*{
(i1)--[fermion](k1),
(i2)--[fermion](k2),
(k1)--[boson](k2),
(k1)--[fermion](k3),
(k2)--[fermion](k4),
(k3)--[boson](k4),
(k3)--[fermion](k5),
(k4)--[fermion](k6),

(k99)--[fermion](k101),
(k100)--[fermion](k102),
(k101)--[boson](k102),
(k101)--[fermion](j1),
(k102)--[fermion](j2),

(i1p)--[fermion,bend right=90](i1),
(k1p)--[fermion](i1p),
(k3p)--[fermion](k1p),
(k5p)--[fermion](k3p),
(k101p)--[fermion](k99p),
(j1p)--[fermion](k101p),
(j1)--[fermion,bend right=90](j1p)
};

\draw [loosely dotted] (k5)--(k99);
\draw [loosely dotted] (k6)--(k100);
\draw [loosely dotted] (k5p)--(k99p);
\draw [loosely dotted] (k6p)--(k100p);
\end{feynman}

\draw[decoration={brace, amplitude=10pt}, decorate] ($(k102.south east)+(10pt,-10pt)$)--($(k2.south west)+(-10pt,-10pt)$) node[pos=0.5,below,yshift=-10pt]{$L$};
\end{tikzpicture}
\caption{Vacuum Feynman diagram obtained by closing the diagram shown in Fig. \ref{fermion-fermion scattering}.}\label{fermion-fermion scattering semivacuum}
\end{figure}
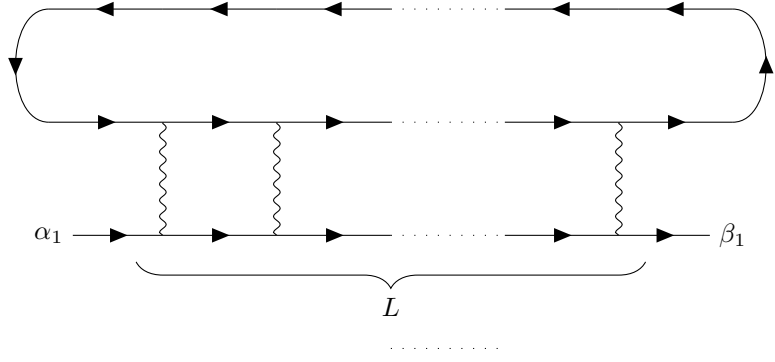

\section{Solutions of the Yang--Baxter equation invariant under the action of $so_{2r}$ in spinor representations}\label{Sec: YBE}

It was shown in \cite{ShanWit,ChDerIs,IsKarKir2} that the operator
\begin{equation}\label{ShanWit_YBE_solution}
\hat{R}(u)=\sum_{k=0}^{\infty}\frac{\hsR_k(u)}{k!}I_k,
\end{equation}
acting on the space $V_\rho\otimes V_\rho$ of the representation $\rho\otimes \rho$ of $so_{2r}$, where $\rho=\Delta_++\Delta_-$ is the spinor representation, satisfies the quantum Yang--Baxter equation in the braid form
\begin{equation}\label{YBEbraid}
\hR_{12}(u)\hR_{23}(u+v)\hR_{12}(v)=\hR_{23}(v)\hR_{12}(u+v)\hR_{23}(u),
\end{equation}
provided that the coefficients $\hsR_k(u)$ in \eqref{ShanWit_YBE_solution} satisfy
\begin{equation}\label{ShanWit_YBE_R_k_conditions}
\hsR_{k+2}(u)=\frac{k+u}{k+2-2r-u}\hsR_k(u).
\end{equation}
The indices $a$ and $b$ in the operators $\hR_{ab}(u)$ in \eqref{YBEbraid} label the tensor factors on which $\hR(u)$ acts non-trivially; for example, -$\hR_{12}(u):=\hR(u)\otimes I$.

Clearly, the relations \eqref{ShanWit_YBE_R_k_conditions} for the coefficients $\hsR_k(u)$ split into two independent families: one for even $k$ and one for odd $k$. Their general solutions can be written in the form \cite{ShanWit,ChDerIs,IsKarKir2}\footnote{The solution obtained in \cite{IsKarKir2} differs from \eqref{R_k explicit} by an overall factor.}
\begin{equation}\label{R_k explicit}
\hsR_{2k}(u)=A(u)(-1)^k\frac{\Gamma(\frac{u}{2}+k)\Gamma(\frac{u}{2}+r-k)}{\Gamma(\frac{u}{2})\Gamma(\frac{u}{2})},\quad
\hsR_{2k+1}(u)=B(u)(-1)^k\frac{\Gamma(\frac{u+1}{2}+k)\Gamma(\frac{u-1}{2}+r-k)}{2\Gamma(\frac{u+1}{2})\Gamma(\frac{u+1}{2})},
\end{equation}
where $A(u)$ and $B(u)$ are arbitrary functions of the spectral parameter $u$. Accordingly, the solution \eqref{ShanWit_YBE_solution} of \eqref{YBEbraid} splits naturally into even and odd parts\footnote{In \cite{ChDerIs}, the even and odd parts of the $R$-operator \eqref{ShanWit_YBE_solution} were denoted by $R^+(u)$ and $R^-(u)$.}:
\begin{equation}\label{ShanWit_sym+antisym}
\hR(u)=\hR^S(u)+\hR^{AS}(u),
\qquad
\hR^S(u)=\sum_{k=0}^\infty \frac{\hsR_{2k}(u)}{(2k)!}I_{2k},
\qquad
\hR^{AS}(u)=\sum_{k=0}^\infty \frac{\hsR_{2k+1}(u)}{(2k+1)!}I_{2k+1}.
\end{equation}

To clarify the meaning of the operators $\hR^{S}(u)$ and $\hR^{AS}(u)$, we introduce the projectors
\begin{align}
\proj^{S}&:=\proj_{++}+\proj_{--}
=\frac{1}{2}(I\otimes I+\Gamma_{2r+1}\otimes \Gamma_{2r+1}),\label{proj^+}\\
\proj^{AS}&:=\proj_{+-}+\proj_{-+}
=\frac{1}{2}(I\otimes I-\Gamma_{2r+1}\otimes \Gamma_{2r+1})\label{proj^-}
\end{align}
from the representation $\rho\otimes \rho=(\Delta_++\Delta_-)\otimes (\Delta_++\Delta_-)$ onto its subrepresentations $\Delta^S$ and $\Delta^{AS}$ defined in \eqref{Delta^S and Delta^AS def}. Recall that $\proj_{\epsilon\epsilon'}:=\proj_\epsilon\otimes \proj_{\epsilon'}$ for $\epsilon,\epsilon'=\pm$, where $\proj_\epsilon$ and $\proj_{\epsilon'}$ are defined in \eqref{rho to Delta_pm projectors}. It is straightforward to verify that
\begin{equation}\label{hR^S and hR^AS projection properties}
\proj^{S}\hR(u)=\hR^{S}(u)=\hR(u)\proj^{S},\qquad
\proj^{AS}\hR(u)=\hR^{AS}(u)=\hR(u)\proj^{AS}.
\end{equation}
Thus, the even part $\hR^{S}(u)$ of the operator $\hR(u)$ acts on the representation space of $\Delta^S$, whereas the odd part $\hR^{AS}(u)$ acts on the representation space of $\Delta^{AS}$.

The solution \eqref{ShanWit_YBE_solution} of equation \eqref{YBEbraid} was obtained in \cite{ShanWit} by successively solving simpler equations known as the $RLL$ relations. $R$-operators acting on the spaces $V_f\otimes V_f$, $V_\rho\otimes V_f$, and $V_\rho\otimes V_\rho$ were constructed, where $V_f$ denotes the space of the defining representation of $so_{2r}$. This approach to solving \eqref{YBEbraid}, however, is rather cumbersome and involves substantial computational difficulties. The final part of the present work is dedicated to deriving the even part $R^{S}(u)$ of this solution by an alternative method proposed in \cite{MacKay} (see also \cite{IsQGLect}, Section~3.13, and \cite{West}, where this method was applied to the construction of $R$-matrices for the exceptional Lie algebras arising from the Freudenthal--Tits magic square). Strictly speaking, this method applies only to the construction of symmetric $R$-matrices acting on tensor squares of irreducible representations. Nevertheless, its extension to the case of the representation $\Delta^{S}$ presents no difficulty. The case of the representation $\Delta^{AS}$ is more involved and is not considered in the present work.

\subsection{Method of solving the Yang--Baxter equation based on properties of Casimir operators}

In this subsection, we consider the general case of a simple Lie algebra $\mathfrak{g}$ and its irreducible representation $T$ acting on the space $V_T$. To formulate the method mentioned above, we use a form of the Yang--Baxter equation alternative to \eqref{YBEbraid}:
\begin{equation}\label{YBE}
R_{12}(u)R_{13}(u+v)R_{23}(v)=R_{23}(v)R_{13}(u+v)R_{12}(u).
\end{equation}
Here $R(u):V_T\otimes V_T\to V_T\otimes V_T$, and the indices $a$, $b$ of the operator $R_{ab}(u)$, as before, label the tensor factors on which $R(u)$ acts non-trivially. Equations \eqref{YBEbraid} and \eqref{YBE} are related by the substitution $\hR(u)=PR(u)$, where $P:V_T\otimes V_T\to V_T\otimes V_T$ denotes the permutation operator.

We assume that the decomposition $T\otimes T=\sum{\lambda} T_\lambda$ of $T\otimes T$ into irreducible components is multiplicity-free. Then any $\mathfrak{g}$-invariant operator $R(u):V_T\otimes V_T\to V_T\otimes V_T$ can be written in the form
\begin{equation}\label{R-general expansion}
R(u)=\sum_{T_\lambda\subseteq T\otimes T}\tau_\lambda(u)\proj_\lambda,
\end{equation}
where $\proj_\lambda$ are mutually orthogonal projectors from $V_T\otimes V_T$ onto the subspaces $V_\lambda$ corresponding to the representations $T_\lambda$, and $\tau_\lambda(u)$ are scalar functions.

\begin{proposition}[\cite{MacKay,IsQGLect}]\label{R-matrix general method}
Let $R(u)$ be the operator given in \eqref{R-general expansion}, which is unitary,
\begin{equation}\label{R-unitarity}
R_{12}(u)R_{21}(-u)=\mathbf{1}\equiv I\otimes I,
\end{equation}
and symmetric,
\begin{equation}\label{R-symmetry}
R_{12}(u)=R_{21}(u),
\end{equation}
where $R_{21}(\pm u):=P R_{12}(\pm u) P$, and suppose that $R(u)$ is a solution of the Yang--Baxter equation \eqref{YBE}. Let $\lambda$ and $\kappa$ be the highest weights of representations $T_\lambda,T_\kappa\subseteq T\otimes T$ such that there exists $X_A\in\mathfrak{g}$ for which
\begin{equation}\label{P_lambda X_a P_kappa}
P_\lambda(I\otimes T(X_A))P_\kappa\neq 0.
\end{equation}
Then the functions $\tau_\lambda(u)$ and $\tau_\kappa(u)$ satisfy
\begin{equation}\label{tau lambda kappa relations}
\frac{\tau_\lambda(u)}{\tau_\kappa(u)}=
\frac{u+\frac{1}{4}(c_{2}^\lambda-c_{2}^\kappa)}
     {u-\frac{1}{4}(c_{2}^\lambda-c_{2}^\kappa)},
\end{equation}
where $c_2^\lambda$ and $c_2^\kappa$ denote the eigenvalues of the quadratic Casimir operator $C_2$ in the representations $T_\lambda$ and $T_\kappa$, respectively.
\end{proposition}

\begin{remark}
The coefficient $\frac{1}{4}$ in the numerator and denominator on the right-hand side of \eqref{tau lambda kappa relations} arises from the fact that in the approach described here the solution of \eqref{YBE} is sought in the form
\begin{equation}\label{R-matrix s.c. limit}
R(u)=\mathbf{1}+\frac{\hC_{T\cdot T}}{u}+O\Big(\frac{1}{u^2}\Big),
\end{equation}
where $\hC_{T\cdot T}$ is the split Casimir operator \eqref{CasDef} of $\mathfrak{g}$ in the representation $T\otimes T$. This choice fixes the freedom of rescaling the spectral parameter by a constant. The ansatz \eqref{R-matrix s.c. limit} is motivated by the fact that its leading term as $u\to\infty$ is the solution $r(u):=\frac{\hC}{u}$ of the quasi-classical Yang--Baxter equation
\begin{equation}
[r_{12}(u),r_{13}(u+v)]
+[r_{13}(u+v),r_{23}(v)]
+[r_{12}(u),r_{23}(v)]=0.
\end{equation}
\end{remark}

\begin{remark}
Strictly speaking, relations \eqref{tau lambda kappa relations} are necessary but not sufficient for the operator \eqref{R-general expansion} to satisfy the Yang–Baxter equation \eqref{YBE}. Nevertheless, in many concrete cases their solution does ensure that the operator \eqref{R-general expansion} satisfies \eqref{YBE} (see, e.g., \cite{IsQGLect}, Section~3.13).
\end{remark}

Let us discuss condition \eqref{P_lambda X_a P_kappa} in greater detail. Since $V_\lambda\subseteq V_T\otimes V_T$ is an invariant subspace, it follows that for any $X_A\in\mfg$, $(I\otimes T(X_A)+T(X_A)\otimes I)\cdot V_\lambda\subseteq V_\lambda$. Therefore, for the orthogonal projectors $\proj_\kappa$ and $\proj_\lambda$ we obtain
\begin{equation}\label{tensor operator def}
\proj_\lambda(I\otimes T(X_A))\proj_\kappa
=\frac{1}{2}\proj_\lambda\big(I\otimes T(X_A)-T(X_A)\otimes I\big)\proj_\kappa.
\end{equation}
The operator $I\otimes T(X_A)-T(X_A)\otimes I$ can be interpreted as a tensor operator of the algebra $\mathfrak{g}$ transforming in the adjoint representation. By the Wigner–Eckart theorem, the matrix \eqref{tensor operator def} is proportional to the Clebsch–Gordan coefficients that intertwine a basis of $V_\lambda$ with a basis of $V_{\ad}\otimes V_\kappa$. Hence, for condition \eqref{P_lambda X_a P_kappa} to hold, it is necessary that
\begin{equation}\label{Wigner-Eckart}
T_\lambda\subseteq \ad\otimes T_\kappa.
\end{equation}

Moreover, note that the invariant projectors $\proj^{\pm}_{T}:=\frac{1}{2}(\mathbf{1}\pm P)$ decompose the representation $T\otimes T$ into its symmetric and antisymmetric parts. Accordingly, each irreducible representation $T_\lambda$ appearing in the decomposition $T\otimes T=\sum_\lambda T_\lambda$ is either symmetric, i.e.\ $\proj^+_{T}\proj_\lambda:=\proj^{+}_\lambda\neq 0$, or antisymmetric, i.e.\ $\proj^-_{T}\proj_\lambda:=\proj^{-}_\lambda\neq 0$. It follows immediately that $\frac{1}{2}\proj_\lambda^\pm\big(I\otimes T(X_A)-T(X_A)\otimes I\big)\proj_\kappa^\pm
=\proj_\lambda^\pm(I\otimes T(X_A))\proj_\kappa^\pm=0$ for all $\lambda$ and $\kappa$. Thus, for condition \eqref{P_lambda X_a P_kappa} to hold, the representations $T_\lambda$ and $T_\kappa$ must have opposite symmetry:
\begin{equation}\label{R-matrix_proj's_symmetry}
V_\lambda\subseteq \proj^+_T V_T^{\otimes 2}\quad \text{and}\quad
V_\kappa\subseteq \proj^-_T V_T^{\otimes 2}
\qquad \text{or}\qquad
V_\lambda\subseteq \proj^-_T V_T^{\otimes 2}\quad \text{and}\quad
V_\kappa\subseteq \proj^+_T V_T^{\otimes 2}.
\end{equation}

We now proceed to apply the method described above to the construction of solutions of the Yang--Baxter equation \eqref{YBE} that are invariant under the action of $so_{2r}$ in the representations $\Delta_\epsilon\otimes \Delta_\epsilon$ for $\epsilon=\pm$.

\subsection{Solutions of the Yang--Baxter equation invariant under the action of the Lie algebra $so_{2r}$ in the representations $\Delta_\pm\otimes \Delta_\pm$}

In accordance with Propositions~\ref{Proposition 6.4.1 from IsRu2} and \ref{R-matrix general method}, we seek solutions of the Yang--Baxter equation \eqref{YBE} in the representation $\Delta_\epsilon\otimes \Delta_\epsilon$ for $\epsilon=\pm$ in the form
\begin{equation}\label{spinor R expansion}
\begin{aligned}
R^{\epsilon\epsilon}(u)&=\tau_0(u)\proj_0^{\epsilon\epsilon}+\tau_2(u)\proj_2^{\epsilon\epsilon}+\dots +\tau_r(u)\proj_r^{\epsilon\epsilon}
\qquad \text{for even }r,\\
R^{\epsilon\epsilon}(u)&=\tau_1(u)\proj_1^{\epsilon\epsilon}+\tau_3(u)\proj_3^{\epsilon\epsilon}+\dots +\tau_r(u)\proj_r^{\epsilon\epsilon}
\qquad \text{for odd }r,
\end{aligned}
\end{equation}
where the projectors $\proj^{\epsilon\epsilon}_k$ are defined in
\eqref{proj^(alpha,pm alpha) for even nu} and
\eqref{proj^(alpha,pm alpha) for odd nu}.
Recall that the image of the projector $\proj_m^{\epsilon\epsilon}$ for $m=0,\dots,r-1$ is the representation space of $T_m$, whereas for $m=r$ it is the representation space of either the self-dual or the anti-self-dual representation $T_r^{+}$ or $T_r^-$, for even and odd $[\frac{r}{2}]$, respectively (see Proposition~\ref{Proposition 6.4.1 from IsRu2}).
It will become clear below that the functions $\tau_k(u)$ for $R^{++}(u)$ and $R^{--}(u)$ coincide; therefore, we do not introduce separate notation for them.

As pointed out in Proposition~\ref{R-matrix general method}, in order to write equations \eqref{tau lambda kappa relations} for $\tau_k(u)$ and $\tau_m(u)$ it suffices to consider only those pairs $\lambda_k$ and $\lambda_m$ (here $\lambda_r$ denotes either $\lambda_r^+$ or $\lambda_r^-$, depending on the parity of $[\frac{r}{2}]$, see Proposition~\ref{Proposition 6.4.1 from IsRu2}) for which conditions \eqref{Wigner-Eckart} and \eqref{R-matrix_proj's_symmetry} hold.

Consider the first of these in greater detail. It is known (see, e.g., \cite{IsRu2}) that the adjoint representation of $so_{2r}$ is described by the Young diagram shown in Fig.~\ref{Young diagrams for ad and T_m} on the left, while the representation $T_m$ is described by the Young diagram shown in Fig.~\ref{Young diagrams for ad and T_m} on the right (in the case $m=r$, this diagram corresponds to the direct sum $T_r=T_r^+\oplus T_r^-$ of the self-dual and anti-self-dual representations). The decomposition of the tensor product $\ad\otimes T_m$ into irreducible components, expressed in terms of Young diagrams and without taking multiplicities into account, is shown in Fig.~\ref{Young diagrams product}. It follows that $T_k$ can appear in $\ad\otimes T_m$ only for $k=m,m\pm 2$. Thus, equations \eqref{tau lambda kappa relations} need only be written for the ratios $\tau_m(u)/\tau_k(u)$ with $m=k+2$, where $k=0,\dots,r-2$.

\begin{figure}
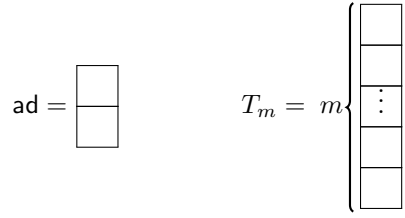

\centering
$\ad=$
\begin{ytableau}
~ \\
~
\end{ytableau}%
\hspace{1.5cm}
$T_m=$\hspace{0.6cm}
\begin{ytableau}
\tikznode{1top}{~}\\
~\\
\vdots\\
~\\
\tikznode{1bottom}{~}
\end{ytableau}
\tikz[overlay,remember picture]{%
\draw[decorate,decoration={brace},thick] ([yshift=-2mm,xshift=-3mm]1bottom.south west) --
([yshift=3.5mm,xshift=-3mm]1top.north west) node[midway,left]{$m$};
}
\caption{Young diagrams describing the representations $\ad$ and $T_m$ of $so_{2r}$.}\label{Young diagrams for ad and T_m}
\end{figure}

\begin{figure}
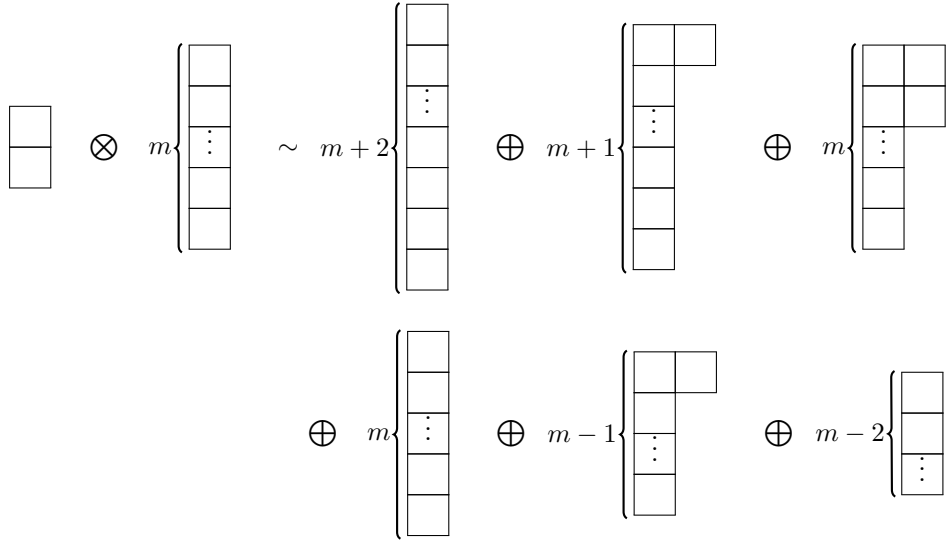

\centering
\begin{ytableau}
~ \\
~
\end{ytableau}%
\hspace{0.5cm}$\bigotimes$\hspace{0.8cm}
\begin{ytableau}
\tikznode{2top}{~}\\
~\\
\vdots\\
~\\
\tikznode{2bottom}{~}
\end{ytableau}
\tikz[overlay,remember picture]{%
\draw[decorate,decoration={brace},thick] ([yshift=-2mm,xshift=-3mm]2bottom.south west) --
([yshift=3.5mm,xshift=-3mm]2top.north west) node[midway,left]{$m$};
}%
\hspace{0.5cm}$\sim$\hspace{1.3cm}
\begin{ytableau}
\tikznode{3top}{~}\\
~\\
\vdots\\
~\\
~\\
~\\
\tikznode{3bottom}{~}
\end{ytableau}
\tikz[overlay,remember picture]{%
\draw[decorate,decoration={brace},thick] ([yshift=-2mm,xshift=-3mm]3bottom.south west) --
([yshift=3.5mm,xshift=-3mm]3top.north west) node[midway,left]{$m+2$};
}%
\hspace{0.5cm}$\bigoplus$\hspace{1.3cm}
\begin{ytableau}
\tikznode{4top}{~} & ~\\
~\\
\vdots\\
~\\
~\\
\tikznode{4bottom}{~}
\end{ytableau}
\tikz[overlay,remember picture]{%
\draw[decorate,decoration={brace},thick] ([yshift=-2mm,xshift=-3mm]4bottom.south west) --
([yshift=3.5mm,xshift=-3mm]4top.north west) node[midway,left]{$m+1$};
}%
\hspace{0.5cm}$\bigoplus$\hspace{0.8cm}
\begin{ytableau}
\tikznode{5top}{~} & ~\\
~ & ~\\
\vdots\\
~\\
\tikznode{5bottom}{~}
\end{ytableau}
\tikz[overlay,remember picture]{%
\draw[decorate,decoration={brace},thick] ([yshift=-2mm,xshift=-3mm]5bottom.south west) --
([yshift=3.5mm,xshift=-3mm]5top.north west) node[midway,left]{$m$};
}

\vspace{0.5cm}
\hspace{3.9cm}$\bigoplus$\hspace{0.8cm}
\begin{ytableau}
\tikznode{5top}{~}\\
~\\
\vdots\\
~\\
\tikznode{5bottom}{~}
\end{ytableau}
\tikz[overlay,remember picture]{%
\draw[decorate,decoration={brace},thick] ([yshift=-2mm,xshift=-3mm]5bottom.south west) --
([yshift=3.5mm,xshift=-3mm]5top.north west) node[midway,left]{$m$};
}%
\hspace{0.5cm}$\bigoplus$\hspace{1.3cm}
\begin{ytableau}
\tikznode{5top}{~} & ~\\
~\\
\vdots\\
\tikznode{5bottom}{~}
\end{ytableau}
\tikz[overlay,remember picture]{%
\draw[decorate,decoration={brace},thick] ([yshift=-2mm,xshift=-3mm]5bottom.south west) --
([yshift=3.5mm,xshift=-3mm]5top.north west) node[midway,left]{$m-1$};
}%
\hspace{0.5cm}$\bigoplus$\hspace{1.3cm}
\begin{ytableau}
\tikznode{6top}{~}\\
~\\
\tikznode{6bottom}{\vdots}
\end{ytableau}
\tikz[overlay,remember picture]{%
\draw[decorate,decoration={brace},thick] ([yshift=-2mm,xshift=-3mm]6bottom.south west) --
([yshift=3.5mm,xshift=-3mm]6top.north west) node[midway,left]{$m-2$};
}%
\caption{Tensor product of the representations $\ad$ and $T_m$ in terms of Young diagrams, without multiplicities.}\label{Young diagrams product}
\end{figure}

Let us now examine condition \eqref{R-matrix_proj's_symmetry} in greater detail.
Proposition~\ref{Proposition 6.4.1 from IsRu2} implies that for each representation
$T_k\subseteq \Delta_\epsilon\otimes \Delta_\epsilon$ there exists exactly one
(up to an overall constant factor) $\mathfrak{g}$-invariant map
$V_{\Delta_\epsilon}\otimes V_{\Delta_\epsilon}\to V_k$.
To construct this map, we introduce the Dirac conjugation of spinors%
\footnote{Here the indices $\alpha$ and $\beta$ refer to the representation $\rho$, not merely to its subrepresentation $\Delta_+$, as in Section~\ref{Sec: diagrams}.}
\begin{equation}
\psi\mapsto \widetilde{\psi}:=\psi^T \mathbf{C}^{-1}
\quad \iff\quad
\psi^\alpha\mapsto \widetilde{\psi}_\alpha
=\psi^\beta (\mathbf{C}^{-1})_{\beta\alpha},
\end{equation}
where $\psi\in V_\rho$ and $\mathbf{C}$ is the charge conjugation matrix.
Restricting this operation to spinors
$\psi_\epsilon,\phi_\epsilon\in V_{\Delta_\epsilon}\subseteq V_\rho$,
we can write the required map in the standard form (see, e.g., \cite{IsRu2}):
\begin{equation}\label{Delta otimes Delta to V_k}
\psi_\epsilon^\alpha\otimes \phi_\epsilon^\beta
\longmapsto
(\widetilde{\psi}_\epsilon)_\alpha
(\Gamma_{[i_1}\dots \Gamma_{i_k]})^\alpha{}_{\beta}
\phi_\epsilon^\beta.
\end{equation}
One readily derives the following transformation properties of the map
\eqref{Delta otimes Delta to V_k} under the permutation
$\psi\leftrightarrow \phi$ for the representations $T_k$ and $T_{k+2}$:
\begin{equation}\label{Projectors symmetries}
\begin{aligned}
(\widetilde{\psi}_\epsilon)_\alpha
(\Gamma_{[i_1}\dots \Gamma_{i_k]})^\alpha{}_{\beta}
\phi_\epsilon^\beta
&=(-1)^{[\frac{r+1}{2}]+\frac{k(k+1)}{2}}
(\widetilde{\phi}_\epsilon)_\alpha
(\Gamma_{[i_1}\dots \Gamma_{i_k]})^\alpha{}_{\beta}
\psi_\epsilon^\beta,\\
(\widetilde{\psi}_\epsilon)_\alpha
(\Gamma_{[i_1}\dots \Gamma_{i_{k+2}]})^\alpha{}_{\beta}
\phi_\epsilon^\beta
&=-(-1)^{[\frac{r+1}{2}]+\frac{k(k+1)}{2}}
(\widetilde{\phi}_\epsilon)_\alpha
(\Gamma_{[i_1}\dots \Gamma_{i_{k+2}]})^\alpha{}_{\beta}
\psi_\epsilon^\beta.
\end{aligned}
\end{equation}
Thus, the representations $T_k$ and $T_{k+2}$ for $k=0,\dots,r-2$ always have opposite symmetry.
Consequently, condition \eqref{R-matrix_proj's_symmetry} imposes no additional constraints on the system of equations \eqref{tau lambda kappa relations} beyond those already implied by condition \eqref{Wigner-Eckart}.

Finally, using the explicit expressions \eqref{c_2 on T_f^(wedge k)} and
\eqref{c_2 on T_pm^(wedge nu)} for the eigenvalues of the quadratic Casimir
operator of $so_{2r}$, we obtain the following form of relations
\eqref{tau lambda kappa relations}:
\begin{equation}\label{tau_k+2 tau_k conditions 1}
\frac{\tau_{k+2}(u)}{\tau_k(u)}
=
\frac{4(r-1)u+(r-1-k)}{4(r-1)u-(r-1-k)}
\end{equation}
for $k=0,2,\dots,r-2$ when $r$ is even and $k=1,3,\dots,r-2$ when $r$ is odd.
Using the freedom of rescaling the spectral parameter,
$u\mapsto \frac{u}{4(r-1)}$, we can rewrite
\eqref{tau_k+2 tau_k conditions 1} in a simpler form
\begin{equation}\label{tau_k+2 tau_k conditions 2}
\frac{\tau_{k+2}(u)}{\tau_k(u)}
=
\frac{u+(r-1-k)}{u-(r-1-k)}.
\end{equation}
Solving equations \eqref{tau_k+2 tau_k conditions 2} successively, starting
from $k=r-2$, and substituting the result into
\eqref{spinor R expansion}, we obtain the following solutions of the
Yang--Baxter equation \eqref{YBE}:
\begin{equation}\label{R-matrix solutions for even r}
R^{\epsilon\epsilon}(u)
=
\proj_r^{\epsilon\epsilon}
+\frac{u-1}{u+1}\proj_{r-2}^{\epsilon\epsilon}
+\frac{(u-1)(u-3)}{(u+1)(u+3)}\proj_{r-4}^{\epsilon\epsilon}
+\dots
+\frac{(u-1)(u-3)\dots(u-(r-1))}
      {(u+1)(u+3)\dots(u+(r-1))}
\proj_0^{\epsilon\epsilon}
\end{equation}
for even $r$, and
\begin{equation}\label{R-matrix solutions for odd r}
R^{\epsilon\epsilon}(u)
=
\proj_r^{\epsilon\epsilon}
+\frac{u-1}{u+1}\proj_{r-2}^{\epsilon\epsilon}
+\frac{(u-1)(u-3)}{(u+1)(u+3)}\proj_{r-4}^{\epsilon\epsilon}
+\dots
+\frac{(u-1)(u-3)\dots(u-(r-2))}
      {(u+1)(u+3)\dots(u+(r-2))}
\proj_1^{\epsilon\epsilon}
\end{equation}
for odd $r$. The overall factor $\tau_r(u)$ in
\eqref{R-matrix solutions for even r} and
\eqref{R-matrix solutions for odd r} has been set equal to one in order to
satisfy the unitarity condition \eqref{R-unitarity}.
Expressions \eqref{R-matrix solutions for even r} and
\eqref{R-matrix solutions for odd r} can be written more compactly as
\begin{equation}\label{R-matrix solutions for general r}
R^{\epsilon\epsilon}(u)
=
\sum_{k=0}^{[\frac{r}{2}]}
\prod_{m=1}^{k}
\frac{u-(2m-1)}{u+(2m-1)}
\cdot
\proj_{r-2k}^{\epsilon\epsilon}.
\end{equation}

As already noted, the solutions $\hR(u)$ and $R(u)$ of the two forms
\eqref{YBEbraid} and \eqref{YBE} of the Yang--Baxter equation are related by $\hR(u)=PR(u)$.
The symmetry properties \eqref{Projectors symmetries} of the representations
$T_k\subseteq \Delta_\epsilon\otimes \Delta_\epsilon$ imply that $P \proj_{r-2k}^{\epsilon\epsilon}
= (-1)^k\proj_{r-2k}^{\epsilon\epsilon}$. Therefore, the operator
\begin{equation}\label{braid R-matrix solutions for general r}
\hR^{\epsilon\epsilon}(u) = \sum_{k=0}^{[\frac{r}{2}]} (-1)^k \prod_{m=1}^{k} \frac{u-(2m-1)}{u+(2m-1)} \cdot  \proj_{r-2k}^{\epsilon\epsilon}
\end{equation}
is a solution of the Yang--Baxter equation in the braid form \eqref{YBEbraid}.

Observe now that the sum of the operators \eqref{braid R-matrix solutions for general r}
for $\epsilon=+$ and $\epsilon=-$ acts on the same space as the operator
$\hR^S(u)$ defined in \eqref{ShanWit_sym+antisym}.
This observation naturally suggests that these operators are proportional, which leads to the following proposition.

\begin{proposition}\label{R^+ and R^++ + R^-- relation}
The solutions $\hR^S(u)$ and $\hR^{++}(u)+\hR^{--}(u)$ of the Yang--Baxter equation \eqref{YBEbraid} in the representation $\Delta^S=\Delta_+^{\otimes 2}\oplus \Delta_-^{\otimes 2}$ of $so_{2r}$ satisfy
\begin{equation}\label{ShanWit_and_our_solutions}
\hR^S(u) = A(u) \prod_{k=0}^{r-1}(u+k) \cdot \bigl(\hR^{++}(u)+\hR^{--}(u)\bigr),
\end{equation}
where $A(u)$ is the function appearing in the definition of $\hR^S(u)$ via \eqref{R_k explicit} and \eqref{ShanWit_sym+antisym}.
\end{proposition}

The proof of this proposition is based on the following technical lemma.

\begin{lemma}\label{Lemma: I_2k proj_r identity}
The invariants $I_{2k}$ defined in \eqref{Ik Def} and the projector $\proj_r^S$ defined in \eqref{proj_k^S and proj_k^AS} satisfy
\begin{equation}\label{I_2k times proj_r}
\frac{I_{2k}}{(2k)!}\proj_r^S = (-1)^k \binom{r}{k} \proj_r^S, \qquad k=0,1,\dots,r.
\end{equation}
\end{lemma}

\begin{proof}
The case $k=0$ is trivial. For $k=1$, note that the invariant $I_2$ is expressed in terms of the split Casimir operator $\hC_\rho$ by \eqref{I2 via hC}. Thus, relation \eqref{I_2k times proj_r} reduces to
\begin{equation}\label{hC_rho times proj_r}
\hC_\rho\proj_r^S = \frac{r}{16(r-1)} \proj_r^S.
\end{equation}
This relation holds because the coefficient of $\proj_r^S$ on the right-hand side of \eqref{hC_rho times proj_r} coincides with the eigenvalue \eqref{hC eigenvalues on Delta times Delta} of the operator $\hC_\rho$ in the representation $T_r$, as required. The remaining cases $k=2,\dots,r$ follow by induction on $k$ using relations \eqref{I2kI2 relation}.
\end{proof}

\begin{proof}[Proof of Proposition \ref{R^+ and R^++ + R^-- relation}]
It was shown in Section \ref{Sec: 3.1} that each invariant $I_{2k}$ can be expressed as a degree-$k$ polynomial in the split Casimir operator $\hC_\rho$ of the algebra $so_{2r}$ in the representation $\rho\otimes \rho$ (see \eqref{I2kI2 relation}). The operator $\hC_\rho$ itself admits a decomposition in terms of the projectors $\proj_k$ onto its eigenspaces given in \eqref{sums of projectors Delta to rho for even r} and \eqref{sums of projectors Delta to rho for odd r} (see also \eqref{Split Casimir expansion}). Therefore, the operator $\hR^S(u)$ defined in \eqref{ShanWit_sym+antisym} as a linear combination of the invariants $I_{2k}$ can likewise be written as a linear combination of these projectors. Since $\hR^S(u)$ acts on the representation space of $\Delta^S$, only those projectors that act on this same space can appear in its decomposition. These are precisely the projectors $\proj_k^S$ introduced in \eqref{proj_k^S and proj_k^AS} for which the parity of $k$ coincides with that of $r$. Thus, the operator $\hR^S(u)$ can be written in the form \eqref{R-general expansion}, with $\proj_k^S$ substituted for $\proj_k^{\epsilon\epsilon}$. Clearly, the functions $\tau_k(u)$ in this case must also satisfy equations \eqref{tau_k+2 tau_k conditions 2}, whose solution is unique up to an overall factor.
It follows that the operators $\hR^S(u)$ and $\hR^{++}(u)+\hR^{--}(u)$ are proportional.

To determine the proportionality coefficient, we multiply $\hR^S(u)$ by the projector $\proj_r^S$ and use Lemma \ref{Lemma: I_2k proj_r identity}, which yields
\begin{equation}\label{hR^S times proj_r}
\hR^S(u)\proj_r^S = A(u) \sum_{k=0}^r \binom{r}{k} \frac{\Gamma(\frac{u}{2}+k)\Gamma(\frac{u}{2}+r-k)} {\Gamma(\frac{u}{2})\Gamma(\frac{u}{2})} \proj_r^S = A(u) \prod_{k=0}^{r-1}(u+k) \proj_r^S,
\end{equation}
where in the second equality we used the standard identity for the rising factorial, $\frac{\Gamma(x+k)}{\Gamma(x)}=x(x+1)\cdots(x+k-1), \qquad x=\frac{u}{2}$, namely
\begin{equation}
\sum_{k=0}^r \binom{r}{k} \frac{\Gamma(x+k)\Gamma(x+r-k)} {\Gamma(x)\Gamma(x)} = \prod_{k=0}^{r-1}(2x+k).
\end{equation}
Multiplying the solution $\hR^{++}(u)+\hR^{--}(u)$ by $\proj_r^S$ gives
\begin{equation}\label{hR^++ + hR^-- times proj_r}
\bigl(\hR^{++}(u)+\hR^{--}(u)\bigr)\proj_r^S = \proj_r^S.
\end{equation}
Comparing \eqref{hR^S times proj_r} with \eqref{hR^++ + hR^-- times proj_r} yields \eqref{ShanWit_and_our_solutions}.
\end{proof}

Thus, we have shown that the method for constructing solutions of the Yang--Baxter equation \eqref{YBEbraid} proposed in \cite{MacKay} indeed makes it possible to obtain the even part of the solution \eqref{ShanWit_YBE_solution} with a minimal amount of computation.

\section*{Conclusion}

In this paper, characteristic identities for the split Casimir operator of the Lie algebra $so_{2r}$ in tensor products of its spinor representations of the same and opposite chiralities have been obtained by two independent methods. On the basis of these identities, projectors onto invariant subspaces of the corresponding representations were constructed and their traces computed.

The resulting expressions made it possible to compute explicitly the colour factor of a ladder Feynman diagram describing the interaction of two fermions in the spinor representation of positive chirality of the group $\Spin(2r)$ mediated by gluon exchange. The results of this work can be useful for calculations in Grand Unified Theories with gauge group $\Spin(10)$.

In addition, using the constructed projectors, solutions of the Yang--Baxter equation in tensor products of spinor representations of the same chirality were obtained. It was shown that the sum of these solutions is proportional to the even part of a previously known solution of the Yang--Baxter equation.

In conclusion, we briefly comment on the possibility of incorporating spinor representations of $so_N$ into a universal description of representations of simple Lie algebras $\mathfrak{g}$. Recall that Vogel’s universality arises from the study of tensor powers of the adjoint representation of $\mfg$. Owing to the existence of a non-degenerate Cartan–Killing metric $\kg_{AB}$, the Lie algebra $\mfg$ in the adjoint representation is naturally embedded in $\mathfrak{so}(\dim \mathfrak{g})$. However, this observation does not appear to be particularly useful for extending the universal description to spinor representations, since these representations do not arise within the adjoint sector of the representation theory of simple Lie algebras.

Another possible way to incorporate spinor representations of $\mathfrak{so}_N$ into a universal framework for Lie algebra representations is based on the observation that, for certain simple (super)algebras, one can choose a basis in which the structure constants are expressed in terms of Dirac gamma matrices (see, e.g., \cite{IsaIv,Stepan} and the references therein).

\section*{Acknowledgements}
 The authors thank A. Sleptsov for useful discussions on the prospects of including spinor representations of $\mathfrak{so}_N$ in the framework of the universal description of Lie algebras.

\appendix

\section{Lie algebra $so_N$ and the Clifford algebra $\Cl_N$}\label{app: so_N and Cl_N props}

In this appendix, we briefly recall the basic properties of the Lie algebra $so_N$, the Clifford algebra $\Cl_N$, and some of their representations that are used in this work (see, e.g., \cite{IsRu1,IsRu2}).

The basis elements $M_{ij}=-M_{ji}$ (where $i,j=1,\dots,N$) of $so_N$ satisfy the commutation relations
\begin{equation}\label{soStructureRelations}
[M_{i_1i_2},M_{j_1j_2}] = X^{k_1k_2}{}_{i_1i_2,j_1j_2} M_{k_1k_2} = \delta_{i_2j_1}M_{i_1j_2} -\delta_{i_2j_2}M_{i_1j_1} -\delta_{i_1j_1}M_{i_2j_2} +\delta_{i_1j_2}M_{i_2j_1}
\end{equation}
with structure constants
\begin{equation}\label{soStructureConstants}
X^{k_1k_2}{}_{i_1i_2,j_1j_2} = \delta_{i_2j_1}\delta^{[k_1}_{i_1}\delta^{k_2]}_{j_2} -\delta_{i_2j_2}\delta^{[k_1}_{i_1}\delta^{k_2]}_{j_1} -\delta_{i_1j_1}\delta^{[k_1}_{i_2}\delta^{k_2]}_{j_2} +\delta_{i_1j_2}\delta^{[k_1}_{i_2}\delta^{k_2]}_{j_1},
\end{equation}
where the square brackets denote antisymmetrisation of indices: $A^{[k_1k_2]}:=\frac{1}{2}\bigl(A^{k_1k_2}-A^{k_2k_1}\bigr)$. The components $\kg_{i_1i_2,j_1j_2}$ of the Cartan--Killing metric and the components $\okg^{i_1i_2,j_1j_2}$ of its inverse (see, e.g., \cite{IsRu1}) are given by
\begin{equation}\label{soKillingMetric}
\kg_{i_1i_2,j_1j_2} = 2(N-2)\bigl(\delta_{i_1j_2}\delta_{i_2j_1}-\delta_{i_1j_1}\delta_{i_2j_2}\bigr), \qquad \okg^{i_1i_2,j_1j_2} = \frac{1}{2(N-2)} \bigl(\delta^{i_1j_2}\delta^{i_2j_1}-\delta^{i_1j_1}\delta^{i_2j_2}\bigr).
\end{equation}

The main objects of interest in this work are the spinor representations of $so_N$.
One of the most convenient tools for studying them is the complex Clifford algebra $\Cl_N$ defined as the associative algebra with unit $I$ and generators $\Gamma_i$, $i=1,\dots,N$, satisfying
\begin{equation}\label{Clifford generators identities}
\Gamma_i\Gamma_j+\Gamma_j\Gamma_i
=
2\delta_{ij}I.
\end{equation}

A basis of $\Cl_N$ consists of the unit element $I$ and all completely antisymmetrised products of the generators $\Gamma_i$:
\begin{equation}\label{ClN basis}
\Gamma_{[i_1\dots i_k]} := \frac{1}{k!} \sum_{\sigma\in S_k} (-1)^{p(\sigma)} \Gamma_{i_{\sigma(1)}} \Gamma_{i_{\sigma(2)}} \dots \Gamma_{i_{\sigma(k)}},
\end{equation}
where $k=1,\dots,N$, and $p(\sigma)$ denotes the parity of the permutation $\sigma\in S_k$.
For convenience, we set the element \eqref{ClN basis} equal to $I$ for $k=0$.
We also note that $\Gamma_{[i_1\dots i_k]}=0$ for all $k>N$.

The standard realisation of the basis elements of $so_N$ satisfying \eqref{soStructureRelations} in terms of the generators $\Gamma_i$ of $\Cl_N$ is
\begin{equation}\label{so via Gammas}
M_{ij}:=\frac{1}{2}\Gamma_{[ij]}=\frac{1}{4}[\Gamma_i,\Gamma_j],
\end{equation}
where $[A,B]:=AB-BA$. From this and the explicit expression \eqref{soKillingMetric} for the inverse
Cartan--Killing metric $\okg^{i_1i_2,j_1j_2}$ of $so_N$, one obtains the following realisation of the split Casimir operator in terms of gamma matrices:
\begin{equation}\label{hC via Gammas}
\hC = \okg^{i_1i_2,j_1j_2} M_{i_1i_2}\otimes M_{j_1j_2} = -\frac{1}{16(N-2)} \Gamma^{[i_1i_2]}\otimes \Gamma_{[i_1i_2]}.
\end{equation}

In the case of even $N=2r$, to which we restrict ourselves in what follows,
the algebra $\Cl_N$ has a unique (up to equivalence) irreducible representation
$\rho$ acting on the space $V_{2^r}$ of dimension $2^r$. Clearly, this representation is restricted to the subalgebra $so_{2r}$ embedded in $\Cl_{2r}$ via \eqref{so via Gammas}.

We define the longest element
\begin{equation}
\Gamma_{2r+1}:=(-i)^r \Gamma_1\cdots\Gamma_{2r}
\end{equation}
of $\Cl_{2r}$, which satisfies the following standard properties:
\begin{equation}\label{Gamma_N+1 properties}
\Gamma_{2r+1}^2=I, \qquad \Gamma_{2r+1}\Gamma_i=-\Gamma_i\Gamma_{2r+1} \quad \text{for} \quad i=1,\dots,2r.
\end{equation}
Using this element, one constructs the invariant mutually orthogonal projectors
\begin{equation}\label{rho to Delta_pm projectors}
\proj_{\pm} := \frac{1}{2} \bigl(I_{2^r}\pm \rho(\Gamma_{2r+1})\bigr),
\end{equation}
which map the representation space of $\rho$ for $\Cl_{2r}$ onto the spaces of the spinor representations $\Delta_\pm$ of $so_{2r}$.

The eigenvalues of the quadratic Casimir operator $C_2$ in the representations $\Delta_\pm$ are computed using \eqref{c_2 via weights}, where the components of the highest weights are $\lambda_{\Delta_\pm} = \Bigl(\frac{1}{2},\dots,\frac{1}{2},\pm\frac{1}{2}\Bigr)$, and the scalar product in the root space is normalised according to \eqref{Root product norm}.
This yields
\begin{equation}\label{c_2 on Delta_pm}
c_2^{\Delta_\pm} = \frac{r(2r-1)}{16(r-1)}.
\end{equation}

For the representations $\Delta_\pm$, one naturally defines the dual (contragredient) representations $\overline{\Delta}_\pm$ by $\oDelta_\pm(M_{ij}) = -\Delta_\pm(M_{ij})^T$, where $T$ denotes matrix transposition. It is known (see, e.g., \cite{IsRu2}) that the following isomorphisms hold:
\begin{equation}\label{Delta bar(Delta) iso}
\begin{aligned}
\text{for even $r$}:\quad &\oDelta_\pm=\Delta_\pm,\\
\text{for odd $r$}:\quad &\oDelta_\pm=\Delta_{\mp}.
\end{aligned}
\end{equation}

\section{Defining representation of $so_N$ and its exterior powers}\label{app: so_N and its external powers}

In presenting an alternative approach to deriving the characteristic identities of the operator $\hC_{\epsilon\epsilon'}$ introduced in \eqref{hC_epsilon epsilon' def}, we use properties of the defining representation of $so_N$ and of its antisymmetrised tensor (exterior) powers. We briefly list these properties below.

The defining representation $T_f$ of $so_N$ is specified on the basis elements $M_{ij}$ by
\begin{equation}\label{soBasis T_f}
T_f(M_{ij})=e_{ij}-e_{ji}
\qquad \iff\qquad
T_f(M_{ij})^a{}_b=\delta^a_i\delta_{jb}-\delta^a_j\delta_{ib},
\end{equation}
where $e_{ij}$ are the matrix units, $(e_{ij})^a{}_b=\delta^a_i\delta_{jb}$. The eigenvalue of the quadratic Casimir operator $C_2$ in this representation follows directly from its definition \eqref{C_2 definition}, the explicit expression for the inverse Cartan-Killing metric \eqref{soKillingMetric} and the basis \eqref{soBasis T_f}, and is given by
\begin{equation}\label{c_2 on T_f}
c_2^{T_f}=\frac{N-1}{2(N-2)}.
\end{equation}

Consider the antisymmetrised tensor product of $k$ defining representations\footnote{Note that, with this notation, $T_1$ coincides with the defining representation $T_f$.}, $T_k:=T_f^{\wedge k}$, acting on the space $V_k:=V_N^{\wedge k}$ of antisymmetric tensors $t^{[i_1,\dots,i_k]} := \sum_{\sigma\in\mathbb{S}_k}(-1)^\sigma t^{i_{\sigma(1)},\dots,i_{\sigma(k)}}$. It is known (see, e.g., \cite{IsRu2}) that for $N=2r$ these representations are irreducible for $k<r$, while for $N=2r+1$ they are irreducible for $k\le r$. Moreover, the representations $T_k$ and $T_{N-k}$ are equivalent for all $k=0,\dots,N$, and henceforth we restrict ourselves to the indices $k=0,\dots,r$. The dimension of $T_k$ is easily computed and equals
\begin{equation}\label{T_f^(wedge k) dimension}
\dim T_k={N\choose k}\equiv \frac{N!}{k!(N-k)!}.
\end{equation}
To compute the corresponding eigenvalue of $C_2$, we use \eqref{c_2 via weights} taking into account that the highest weight $\lambda_k$ of $T_k$ in the orthogonal basis $\{e^{(i)}\}$ of the root space of $so_N$ has components $\lambda_k=(\underbrace{1,\dots,1}_{k},0,\dots,0)$, and that the Weyl vector has components $\delta=(r-1,r-2,\dots,1,0)$ for $N=2r$ and $\delta=(r-\frac{1}{2},r-\frac{3}{2},\dots,\frac{3}{2},\frac{1}{2})$ for $N=2r+1$ (see, e.g., \cite{IsRu2}). With the normalisation \eqref{Root product norm} of the scalar product on the root space of $so_N$, we obtain
\begin{equation}\label{c_2 on T_f^(wedge k)}
c_2^{T_k}=\frac{k(N-k)}{2(N-2)}.
\end{equation}

For $N=2r$, the representation $T_r$ of $so_{2r}$ is reducible and decomposes into the so-called self-dual and anti-self-dual representations $T_r^\pm$ acting on the spaces $V_{2r}^\pm$ of antisymmetric tensors $t_\pm^{[i_1,\dots,i_r]}$ satisfying the additional conditions (see, e.g., \cite{IsRu2})
\begin{equation}
\frac{(-i)^r}{r!}\varepsilon_{i_1\dots i_r k_1\dots k_r}\, t_{(\pm)}^{k_1\dots k_r}
=
\pm\, t_{(\pm)i_1\dots i_r},
\end{equation}
where $\varepsilon_{i_1\dots i_r k_1\dots k_r}$ is the totally antisymmetric tensor with $\varepsilon_{1,2,\dots,2r}=1$, and indices are raised and lowered using the metric $\delta_{ij}$. The dimensions of these representations are
\begin{equation}\label{T_(pm)^(wedge nu) dimension}
\dim T_r^\pm=\frac{1}{2}{2r \choose r}\equiv \frac{(2r)!}{2(r!)^2}.
\end{equation}
The corresponding eigenvalues of the quadratic Casimir operator are obtained from \eqref{c_2 via weights}, using that the highest weights of $T_r^\pm$ for $so_{2r}$ have components $\lambda_r^\pm=(1,\dots,1,\pm 1)$ (see, e.g., \cite{FultonHarris}) and the normalisation \eqref{Root product norm} of the scalar product on the root space of $so_{2r}$:
\begin{equation}\label{c_2 on T_pm^(wedge nu)}
c_2^{T_r^\pm}=\frac{r^2}{4(r-1)}.
\end{equation}

\end{document}